\documentclass[11pt]{article}
\usepackage{color,amsmath, amsthm, amsfonts, amssymb, dsfont, graphicx, listings, graphics, epsfig, enumerate, bbm,dsfont,setspace, natbib
}

\usepackage[pdftex]{hyperref}
\usepackage[margin=1in]{geometry}
\usepackage[fleqn,tbtags]{mathtools}
\usepackage[english]{babel}

\title{Set-valued shortfall and divergence risk measures}
\author{\c{C}a\u{g}{\i}n Ararat\thanks{Department of Industrial Engineering, Bilkent University, Ankara, 06800, Turkey, cararat@bilkent.edu.tr}
\and
Andreas H. Hamel\thanks{School for Economics and Management, Free University Bozen, Bozen-Bolzano, 39031, Italy, andreas.hamel@unibz.it}
\and
Birgit Rudloff\thanks{Institute for Statistics and Mathematics, Vienna University of Economics and Business, Vienna, 1020, Austria, brudloff@wu.ac.at}
}
\date{September 10, 2017}

\addto\captionsenglish {}
\makeatletter \renewenvironment{proof}[1][\proofname] {\par\pushQED{\qed}\normalfont\topsep6\p@\@plus6\p@\relax\trivlist\item[\hskip\labelsep\bfseries#1\@addpunct{.}]\ignorespaces}{\popQED\endtrivlist\@endpefalse} \makeatother

\newtheorem{thm}{Theorem}[section]
\newtheorem{cor}{Corollary}[section]
\newtheorem{lem}{Lemma}[section]
\newtheorem{prop}{Proposition}[section]
\newtheorem{defn}{Definition}[section]
\newtheorem{example}{Example}[section]
\newtheorem{rem}{Remark}[section]

\newtheorem{assumption}{Assumption}[section]

\numberwithin{equation}{section}

\let\abs=\envert

\renewcommand{\o}{\omega}
\renewcommand{\O}{\Omega}
\renewcommand{\Pr}{\mathbb{P}}
\newcommand{\E}{\mathbb{E}}
\newcommand{\Q}{\mathbb{Q}}
\renewcommand{\P}{\mathcal{P}}
\newcommand{\G}{\mathcal{G}}
\newcommand{\M}{\mathcal{M}}
\newcommand{\F}{\mathcal{F}}
\newcommand{\B}{\mathcal{B}}
\newcommand{\C}{\mathcal{C}}
\newcommand{\D}{\mathcal{D}}
\newcommand{\K}{\mathcal{K}}
\newcommand{\T}{\mathbb{T}}

\DeclareMathOperator{\ent}{ent}

\DeclareMathOperator{\cl}{cl}
\DeclareMathOperator{\co}{co}
\DeclareMathOperator{\cone}{cone}
\DeclareMathOperator{\convex}{convex}

\DeclareMathOperator{\graph}{graph}

\DeclareMathOperator{\interior}{int}

\DeclareMathOperator{\market}{mar}

\DeclareMathOperator{\dom}{dom}

\DeclareMathOperator{\esssup}{ess\,sup}
\DeclareMathOperator{\essinf}{ess\,inf}
\DeclareMathOperator{\Avar}{AV@R}

\newcommand{\R}{\mathrm{I\negthinspace R}}
\newcommand{\N}{\mathrm{I\negthinspace N}}

\newcommand{\bs}{\backslash}

\newcommand{\of}[1]{\ensuremath{\left( #1 \right)}}
\newcommand{\cb}[1]{\ensuremath{ \left\{ #1 \right\} }}
\newcommand{\sqb}[1]{\ensuremath{ \left[ #1 \right] }}

\def\prehp(#1,#2){\ensuremath{  #1 \cdot #2 }}

\begin{document}
\maketitle
\thispagestyle{empty}

\begin{abstract}
Risk measures for multivariate financial positions are studied in a utility-based framework. Under a certain incomplete preference relation, shortfall and divergence risk measures are defined as the optimal values of specific set minimization problems. The dual relationship between these two classes of multivariate risk measures is constructed via a recent Lagrange duality for set optimization. In particular, it is shown that a shortfall risk measure can be written as an intersection over a family of divergence risk measures indexed by a scalarization parameter. Examples include set-valued versions of the entropic risk measure and the average value at risk. As a second step, the minimization of these risk measures subject to trading opportunities is studied in a general convex market in discrete time. The optimal value of the minimization problem, called the market risk measure, is also a set-valued risk measure. A dual representation for the market risk measure that decomposes the effects of the original risk measure and the frictions of the market is proved.
\\
\\[-5pt]
\textbf{Keywords and phrases: }Optimized certainty equivalent; shortfall risk; divergence; relative entropy; entropic risk measure; average value at risk; set-valued risk measure; multivariate risk; incomplete preference; transaction cost; solvency cone; liquidity risk; infimal convolution; Lagrange duality; set optimization.\\
\\[-5pt]
\textbf{Mathematics Subject Classification (2010): }91B30, 46N10, 46A20, 26E25, 90C46.
\end{abstract}

\section{Introduction}

Risk measures for random vectors have recently gained interest in the financial mathematics community. Introduced in the pioneering work \cite{jouini}, \emph{set-valued risk measures} have been used to quantify \emph{financial risk} in markets with frictions such as transaction costs or illiquidity effects. These risk measures are functions which assign to an $m$-dimensional random vector $X$ a set $R(X)\subseteq\R^m$ whose elements can be used as risk compensating portfolios. Here, $X$ denotes a financial position in $m$ assets whose components are in terms of physical units rather than values with respect to a specific num\'{e}raire. More recently, set-valued risk measures have also been used to quantify \emph{systemic risk} in financial networks; see \cite{FRW15}, \cite{AR15}. In this case, $m$ is the number of financial institutions and the components of $X$ denote the corresponding magnitudes of a random shock (equity/loss) for these institutions.

The coherent set-valued risk measures in \cite{jouini} have been extended to the convex case in \cite{hh:duality} and to random market models in \cite{hhr:setval}. These extensions were possible by an application of the duality theory and, in particular, the Moreau-Fenchel biconjugation theorem for set-valued functions developed in \cite{andreasduality}. Extensions to the dynamic framework have been studied in \cite{FR13,FR13b,FRsurvey}, \cite{LBT13} and to set-valued portfolio arguments in \cite{CM13}. Scalar risk measures for multivariate random variables, which can be interpreted as scalarizations of set-valued risk measures (see \citealt[Section~2.4]{FRsurvey}) have been studied in \cite{JK95}, \cite{BR06}, \cite{Weber13} (financial risk) as well as in \cite{CIM13}, \cite{BFFMb15} (systemic risk).

Set-valued generalizations of some well-known scalar coherent risk measures have already been studied such as the set-valued version of the average value at risk in \cite{hry:avar}, \cite{FR13b}, \cite{HLR13}, or the set of superhedging portfolios in markets with transaction costs in \cite{hhr:setval}, \cite{LR13}, \cite{FR13}. Other examples of coherent risk measures for multivariate claims can be found in \cite{B06}, \cite{CM13}. To the best of our knowledge, apart from superhedging with certain trading constraints in markets with frictions, which leads to set-valued convex risk measures (see \citealt{HLR13}), no other examples have been studied in the convex case yet.

This paper introduces utility-based convex risk measures for random vectors. The basic assumption is that the investor has a complete risk preference towards each asset which has a numerical representation in terms of a von Neumann - Morgenstern loss (utility) function. However, her risk preference towards multivariate positions is incomplete and it can be represented in terms of the vector of individual loss functions. Based on this incomplete preference, the \emph{shortfall risk} of the random vector $X$ is defined as the collection of all portfolios $z\in\R^m$ for which $X+z$ is preferred to a benchmark portfolio $z^0\in\R^m$. As an example, when the individual loss functions are exponential, we obtain set-valued versions of the well-known \emph{entropic risk measure} (see \citealt{fs:srm, fs:sf}).

We formulate the computation of the shortfall risk measure as a constrained set optimization problem and apply recent tools from the set optimization literature to obtain a dual formulation. In particular, using the Lagrange duality in \cite{hl:lagrange}, another type of convex risk measures, called \emph{divergence risk measures}, are obtained in the dual problem. A divergence risk measure is defined based on the trade-off between consuming a deterministic amount $z\in\R^m$ of the position today and realizing the expected loss of the remaining amount $X-z$ at terminal time. The decision making problem is bi-objective: The investor wants to choose a portfolio $z$ so as to \emph{maximize} $z$ and \emph{minimize} the (vector-valued) expected loss due to $X-z$ at the same time, both of which are understood in the sense of set optimization (see Section~\ref{divergence-sec}). The two objectives are combined by means of a relative weight (scalarization) parameter $r\in\R^m_{+}$ and the divergence risk of $X$ is defined as an unconstrained set optimization problem over the choices of the deterministic consumption $z$. As special cases, we obtain the definition of the set-valued average value at risk given in \cite{hry:avar} as well as a convex version of it.

One of the main results of this paper is that a shortfall risk measure can be written as the intersection of divergence risk measures indexed by their relative weights and, in general, the intersection is not attained by a unique relative weight. Hence, the shortfall risk measure is a (much) more conservative risk measure than a divergence risk measure. While the shortfall risk measure is more difficult to compute as a constrained optimization problem, we show that the computation of a divergence risk measure can be reduced to the computation of scalar divergence risk measures (\emph{optimized certainty equivalents} in \citealt{oldoce, bt:oce}). On the flip side, and in contrast to shortfall risk measures, to be able to use a divergence risk measure, the decision maker has to specify the relative weight of her loss with respect to her consumption for each asset.

While shortfall and divergence risk measures are defined based on the preferences of the investor, they do not take into account how the market frictions affect the riskiness of a position. In Section~\ref{market}, we propose a method for incorporating these frictions in the computation of risk. We generalize the notion of \emph{market risk measure} (see \citealt{hry:avar} with the name \emph{market-extension}) by including trading constraints modeled by convex random sets, and considering issues of liquidation into a certain subcollection of the assets. In contrast to \cite{hry:avar}, we allow for a convex (and not necessarily conical) market model to include temporary illiquidity effects in which the bid-ask prices depend on the magnitude of the trade, and thus, are given by the shape of the limit order book; see \cite{AT07}, \cite{PP10}, for instance. Letting $R$ be a (market-free) risk measure such as a shortfall or divergence risk measure, its induced market risk measure is defined as the minimized value of $R$ over the set of all financial positions that are attainable by trading in the market. As the second main result of the paper, we prove a dual representation theorem for the market risk measure (Theorem~\ref{liquidationtheorem}). In particular, we show that the penalty (Fenchel conjugate) function of the market risk measure is the sum of the penalty function of the base risk measure $R$ and the supporting halfspaces of the convex regions of the market.

The rest of this paper is organized as follows. In Section~\ref{scalartheory}, we review the scalar theory of shortfall and divergence risk measures. However, we generalize the standard results in the literature as we allow for extended real-valued loss functions and we do not impose any growth conditions on the loss functions as in \cite{fs:srm}, \cite{bt:oce}. The main part of the paper is Section~\ref{mainsection}, where set-valued shortfall and divergence risk measures are studied. In Section~\ref{examples}, set-valued entropic risk measures are studied as examples of shortfall risk measures and set-valued average value at risks are recalled as examples of divergence risk measures. Market risk measures in a general convex market model with liquidation and trading constraints are studied in Section~\ref{market}. All proofs are collected in Section~\ref{appendix}.

\section{Scalar shortfall and divergence risk measures}\label{scalartheory}

In this section, we summarize the theory of (utility/loss-based) shortfall and divergence risk meas\-ures for univariate financial positions. Shortfall risk measures are introduced in \cite{fs:srm}. Divergence risk measures are introduced in \cite{oldoce}, and analyzed further in \cite{bt:oce} with the name \emph{optimized certainty equivalents} and in \cite{CKpre} with the name \emph{divergence utilities} for their negatives. The dual relationship between shortfall and divergence risk measures is pointed out in \cite{schied} and \cite{bt:oce}. In terms of the assumptions on the underlying loss function, we generalize the results of these papers by dropping growth conditions; see Section~\ref{technical} for a comparison.

The proofs of the results of this section are given in Section~\ref{scalarproof} and most of them inherit the convex duality arguments in \cite{bt:oce} rather than the analytic arguments in \cite{fs:srm}.

\begin{defn}\label{loss}
A lower semicontinuous, convex function $f:\R\rightarrow\R\cup\{+\infty\}$ with effective domain $\dom f =\cb{x\in\R\mid f(x)<+\infty}$ is said to be a \emph{loss function} if it satisfies the following properties:
\begin{enumerate}[(i)]
\item $f$ is nondecreasing with $\inf_{x\in\R} f(x)>-\infty$.
\item $0\in \dom f$.
\item $f$ is not identically constant on $\dom f$.
\end{enumerate}
\end{defn}

Throughout this section, let $\ell \colon \R\rightarrow\R\cup\{+\infty\}$ be a loss function. Definition~\ref{loss} above guarantees that $\interior\ell(\R)\not=\emptyset$, where $\interior$ denotes the interior operator. Let us fix a threshold level $x^{0}\in\interior\ell(\R)$ for expected loss values. Without loss of generality, we assume $x^0 = 0$. Based on the loss function $\ell$, we define the shortfall risk measure on the space $L^\infty$ of essentially bounded real-valued random variables of a probability space $(\O,\mathcal{F},\Pr)$, where random variables are identified up to almost sure equality.

\begin{defn}\label{srm-def}
The function $\rho_{\ell} \colon L^{\infty}\rightarrow \R\cup\{\pm\infty\}$ defined by
\begin{equation}\label{srm-def2}
\rho_{\ell}(X)=\inf\{s\in\R\mid\mathbb{E}\sqb{\ell(-X-s)}\leq 0\}
\end{equation}
is called the \emph{shortfall risk measure}.
\end{defn}

\begin{prop}\label{srm-lsc}
The function $\rho_{\ell}$ is a (weak$^{*}$-)lower semicontinuous convex  risk measure in the sense of \citet[Definitions~4.1, 4.4]{fs:sf}. In particular, $\rho_\ell$ takes values in $\R$.
\end{prop}

\begin{rem}
Since $\inf_{x\in\R}\ell(x)>-\infty$, it holds $\E\sqb{\ell(-X-s)}>-\infty$ for every $X\in L^\infty, s\in\R$. Hence, the expectation in \eqref{srm-def2} is always well-defined. Moreover, the assumption $x^0=0\in\interior \ell(\R)$ is essential for the finiteness of $\rho_{\ell}(X)$ as shown in the proof of Proposition~\ref{srm-lsc}; see Section~\ref{scalarproof}.
\end{rem}

According to Definition~\ref{srm-def}, the number $\rho_{\ell}(X)$ can be seen as the optimal value of a convex minimization problem. The next proposition computes $\rho_{\ell}(X)$ as the optimal value of the corresponding Lagrangian dual problem. Its proof in Section~\ref{scalarproof} is an easy application of strong duality.

\begin{prop}\label{srm-simple-dual}
For every $X\in L^{\infty}$,
\begin{equation}\label{srm-dual}
\rho_{\ell}(X)=\sup_{\lambda\in\R_+}\delta_{\ell,\lambda}(X),
\end{equation}
where
\begin{align}\label{divergencerisk}
\delta_{\ell,\lambda}(X)\coloneqq & \inf_{s\in\R\colon\E\sqb{\ell(-X-s)}<+\infty}\of{s+\lambda\E\sqb{\ell(-X-s)}}\\
&=\begin{cases}\inf_{s\in\R}\of{s+\lambda\mathbb{E}\sqb{\ell(-X-s)}}& \text{if }\lambda>0,\\ -\essinf X - \sup \dom\ell &\text{if }\lambda=0.\end{cases}\notag
\end{align}
\end{prop}

Note that $\delta_{\ell,\lambda}$ is a monotone and translative function on $L^\infty$ for each $\lambda\in\R_+$. Our aim is to determine the values of $\lambda$ for which this function is a lower semicontinuous convex risk measure with values in $\R$. To that end, we define the Legendre-Fenchel conjugate $g \colon \R \to \R\cup\{\pm\infty\}$ of the loss function by
\begin{equation}
g(y) := \ell^*(y) =\sup_{x\in\R}(xy-\ell(x)).
\end{equation}
In the following, we will adopt the convention $(+\infty) \cdot 0 = 0$ as usual in convex analysis, see \cite{rockafellar2}. We will also use $\frac{1}{+\infty} = 0$ as well as $\frac{1}{0} = +\infty$.

\begin{defn}\label{divergence}
A proper, convex, lower semicontinuous function $\varphi \colon \R\rightarrow\R\cup\{+\infty\}$ with effective domain $\dom \varphi=\cb{y\in\R\mid \varphi(y)<+\infty}$ is said to be a \emph{divergence function} if it satisfies the following properties:
\begin{enumerate}[(i)]
\item $0\in \dom\varphi \subseteq\R_{+}$.
\item $\varphi$ attains its infimum over $\R$.
\item $\varphi$ is not of the form $y\mapsto +\infty\cdot 1_{\{y<0\}}+(ay+b)\cdot 1_{\{y\geq 0\}}$ with $a\in\R_{+}\cup\{+\infty\}$ and $b\in\R$.
\end{enumerate}
\end{defn}

\begin{prop}\label{conjugation-bijection}
Legendre-Fenchel conjugation furnishes a bijection between loss and divergence functions.
\end{prop}

\begin{rem}\label{divergenceindex}
Let $\lambda>0$. If $f$ is a loss function, then $\lambda f$ is also a loss function. If $\varphi$ is a divergence function, then the function $y\mapsto\varphi_{\lambda}(y)\coloneqq\lambda\varphi(\frac{y}{\lambda})$ on $\R$ is also a divergence function. The functions $f$ and $\varphi$ are conjugates of each other if and only if $\lambda f$ and $\varphi_{\lambda}$ are. In this case, we also define the \emph{recession function} $\varphi_0\colon\R\to\R\cup\cb{+\infty}$ of $\varphi$ by
\begin{equation}
\varphi_0(y)\coloneqq \sup_{\lambda>0}\of{\varphi_\lambda(y)-\lambda \varphi(0)}=\lim_{\lambda\downarrow 0}\varphi_\lambda(y)=\begin{cases}y\sup\dom f & \text{if }y\geq 0,\\ +\infty & \text{if }y<0,\end{cases}
\end{equation}
for each $y\in\R$. Here, $\lambda\mapsto \varphi_\lambda(y)-\lambda g(0)$ is a nonincreasing convex function on $\R_{++}$ for each $y\in\R$. Moreover, the second equality holds thanks to the assumption $0\in\dom \varphi$, see \citet[Theorem~8.5, Corollary~8.5.2]{rockafellar}. The last equality is due to the fact that the support function of the effective domain of the proper convex function $f$ coincides with the recession function $\varphi_0$ of its conjugate, see \citet[Theorem~13.3]{rockafellar}.
\end{rem}

We next recall the notion of divergence. To that end, let $\M(\Pr)$ be the set of all probability measures on $(\O,\mathcal{F})$ that are absolutely continuous with respect to $\Pr$.

\begin{defn}\label{divergencedefn}
Let $\varphi$ be a divergence function with the corresponding loss function $f$. For $\lambda\in\R_+$ and $\Q\in\M(\Pr)$, the quantity
\begin{equation}
I_{\varphi,\lambda}(\Q\mid\Pr) \coloneqq \E\sqb{\varphi_{\lambda}\of{\frac{d\Q}{d\Pr}}} = \begin{cases}\lambda\E\sqb{\varphi\of{\frac{1}{\lambda}\frac{d\Q}{d\Pr}}}&\text{if }\lambda>0,\\ \sup\dom f & \text{if }\lambda=0\end{cases}
\end{equation}
is called the $(\varphi,\lambda)$-\emph{divergence} of $\Q$ with respect to $\Pr$.
\end{defn}

\begin{rem}
$I_{\varphi,1}$ is the usual $\varphi$-divergence in the sense of \cite{csi}. It is a notion of ``distance" between probability measures and includes the well-known \emph{relative entropy} as a special case, see \eqref{relent} below.
\end{rem}

Note that $g = \ell^*$ is a divergence function, and $\dom g$ is an interval of the form $[0,\beta)$ or $[0,\beta]$ for some $\beta\in\R_{++}\cup\{+\infty\}$. Here, we have $\dom g\not=\{0\}$ since otherwise $g$ would be of the form $y\mapsto +\infty\cdot 1_{\{y<0\}}+(ay+b)\cdot 1_{\{y\geq 0\}}$ for $a=+\infty$ and $b=g(0)$. For each $\lambda>0$, $y\mapsto g_{\lambda}(y)\coloneqq\lambda g(\frac{y}{\lambda})$ on $\R$ is a divergence function with $\dom g_\lambda=[0,\lambda\beta)$ or $\dom g_\lambda=[0,\lambda\beta]$ by Remark~\ref{divergenceindex}, and the corresponding $(g,\lambda)$-divergence is defined according to Definition~\ref{divergencedefn}. In the case $\lambda=0$, $y\mapsto g_0(y)=+\infty\cdot1_{\cb{y<0}}+(\sup\dom\ell)y\cdot 1_{\cb{y\geq 0}}$ on $\R$ is not a divergence function. Moreover, we have $\dom g_0=\cb{0}$ if $\dom\ell=\R$, and $\dom g_0=\R_+$ if $\dom \ell\neq\R$.

\begin{thm}\label{oce-duality}
For every $\lambda\in\R_+$ and $X\in L^{\infty}$,
\begin{equation}\label{oce-primal-dual}
\delta_{\ell,\lambda}(X) = \sup_{\Q\in\M(\Pr)} \of{\E^{\Q}\sqb{-X}-I_{g,\lambda}(\Q\mid\Pr)}.
\end{equation}
Moreover, $\delta_{\ell,\lambda}$ is a lower semicontinuous convex risk measure if $1\in \dom g_{\lambda}$, and $\delta_{\ell,\lambda}(X)=-\infty$ for every $X\in L^\infty$ otherwise. Hence,
\begin{equation}
\rho_{\ell}(X)=\sup_{\lambda\in\R_+\colon 1\in \dom g_{\lambda}}\delta_{\ell,\lambda}(X).
\end{equation}
In particular, if $\dom\ell=\R$, then
\begin{equation}
\rho_{\ell}(X)=\sup_{\lambda>0\colon \frac{1}{\lambda}\in \dom g}\delta_{\ell,\lambda}(X).
\end{equation}
\end{thm}

\begin{defn}\label{divergencermdefn}
For $\lambda\in\R_+$ with $1\in \dom g_\lambda$, the function $\delta_{\ell,\lambda}: L^{\infty}\rightarrow\R$ is called the \emph{divergence risk measure} with weight $\lambda$.
\end{defn}

In \eqref{oce-primal-dual}, note that a divergence risk measure is represented in terms of probability measures. More generally, by \citet[Theorem 4.33]{fs:sf}, every lower semicontinuous convex risk measure $\rho \colon L^\infty \to \R$ has a \emph{dual representation} in the sense that it is characterized by its so-called \emph{penalty function} $\alpha_{\rho} \colon \M(\Pr)\to\R\cup\{+\infty\}$ by the following relationships:
\begin{equation}\label{dual-srm}
\rho(X)=\sup_{\Q\in\M(\Pr)}\of{\E^{\Q}\sqb{-X}-\alpha_{\rho}(\Q)},\quad \alpha_{\rho}(\Q) = \sup_{X\in L^{\infty}}\of{\E^{\Q}\sqb{-X}-\rho(X)}.
\end{equation}
In Proposition~\ref{divergence-minimal}, we check that \eqref{oce-primal-dual} is indeed the dual representation of the divergence risk measure $\delta_{\ell,\lambda}$. We also compute the penalty function of the shortfall risk measure in terms of the penalty functions of divergence risk measures.

\begin{prop}\label{divergence-minimal}
Let $\lambda\in\R_+$ with $1\in \dom g_\lambda$. For each $\Q\in\M(\Pr)$, it holds
\begin{equation}
\alpha_{\delta_{\ell,\lambda}}(\Q)=I_{g,\lambda}(\Q\mid\Pr),
\end{equation}
and moreover,
\begin{equation}
\alpha_{\rho_{\ell}}(\Q)=\inf_{\lambda\in\R_+}I_{g,\lambda}(\Q\mid\Pr)=\inf_{\lambda\in\R_+\colon 1\in\dom g_\lambda}\alpha_{\delta_{\ell,\lambda}}(\Q).
\end{equation}
\end{prop}

\section{Set-valued shortfall and divergence risk measures}\label{mainsection}

In this section, we introduce utility-based shortfall and divergence risk measures for multivariate financial positions, the central objects of this paper. The proofs are presented in Section~\ref{setvaluedproof}.

Let us introduce some notation that will be used frequently throughout the rest of the paper. Let $m\geq 1$ be an integer and $\abs{\cdot}$ an arbitrary fixed norm on $\R^m$. By $\R^m_+$ and $\R^m_{++}$, we denote the set of elements of $\R^m$ with nonnegative and strictly positive components, respectively.

Throughout, we consider a probability space $(\O,\mathcal{F},\Pr)$. We denote by $L_{m}^{0} := L^0_m(\Omega,\mathcal{F},\mathbb{P})$ the linear space of random variables taking values in $\R^{m}$, where two elements are identified if they are equal $\mathbb{P}$-almost surely; and we define
\begin{align}
& L_{m}^{1}=\{X\in L_{m}^{0}\mid\mathbb{E}\sqb{\abs{X}}<+\infty\},\quad L_{m}^{\infty}=\{X\in L_{m}^{0}\mid\esssup\abs{X}<+\infty\}, \\
& L^{p}_{m,+}=\{X\in L_{m}^{p}\mid \mathbb{P}\{X\in\R_{+}^{m}\}=1\},\quad p\in\{1,+\infty\}.\notag
\end{align}

Componentwise ordering of vectors is denoted by $\leq$, that is, for $x,z\in\R^m$, it holds $x\leq z$ if and only if $x_i\leq z_i$ for each $i\in\cb{1,\ldots,m}$. The Hadamard product of $x,z\in\R^m$ is defined by $x\cdot z \coloneqq (x_1z_1,\ldots,x_mz_m)^{\mathsf{T}}$. We denote by $\mathcal{P}(\R^{m})$ the power set of $\R^m$, that is, the set of all subsets of $\R^m$ including the empty set $\emptyset$. On $\mathcal{P}(\R^{m})$, the Minkowski addition and multiplication with scalars are defined by $A+B=\{a+b\mid a\in A,\; b\in B\}$ and $sA=\{sa\mid a\in A\}$ for $A, B \subseteq \R^{m}$ and $s \in \R$ with the conventions $A+\emptyset=\emptyset+B=\emptyset+\emptyset=\emptyset$, $s\emptyset=\emptyset$ $(s\not=0)$, and $0\emptyset=\{0\}\subseteq\R^{m}$. We also use the shorthand notations $A-B = A + (-1)B$ and  $z+A = \{z\}+A$. For $x\in\R^m$ and a nonempty set $A\subseteq\R^m$, we set $x\cdot A\coloneqq \cb{x\cdot a\mid a\in A}$. These operations can be defined on the power set $\mathcal{P}(L_{m}^{p})$ of $L_{m}^{p}$, $p\in\{0,1,\infty\}$, in a similar way. (In)equalities between random variables are understood in the $\Pr$-almost sure sense.

\subsection{The incomplete preference relation}\label{incomplete-sec}

Let $m \geq 1$ be an integer denoting the number of assets in a financial market. The linear space $\R^{m}$ is called the space of eligible portfolios. This means that every $z\in \R^{m}$ is a potential deposit to be used at initial time in order to compensate for the risk of a financial position.

We model a financial position as an element $X \in L_{m}^{\infty}$, where $X_i(\omega)$ represents the number of physical units in the $i^{\text{th}}$ asset for $i\in\{1, \ldots, m\}$ when the state of the world $\omega\in\Omega$ occurs. We assume that the investor has a (possibly) incomplete preference relation for multivariate financial positions in $L_m^\infty$. Its numerical representation is in terms of the individual loss functions for the assets and a comparison rule for the vectors of expected losses:
\begin{enumerate}[(i)]
\item \textbf{Loss functions for assets: }We assume that the investor has a complete preference relation $\succeq_i$ on $L^\infty$ corresponding to each asset $i\in\{1,\ldots,m\}$ and this preference relation has a von Neumann - Morgenstern representation given by a (scalar) loss function $\ell_{i} \colon \R\rightarrow\R\cup\{+\infty\}$ (see Definition~\ref{loss}). That is, for $X_i, Z_i\in L^\infty$,
\begin{equation}
X_i \succeq_i Z_i \;\Leftrightarrow\; \E\sqb{\ell_i(-X_i)}\leq \E\sqb{\ell_i(-Z_i)}.
\end{equation}

\item \textbf{Vector loss function: }Let $\ell \colon \R^m \to \R^m\cup\{+\infty\}$ be the \emph{vector loss function} defined by
\begin{equation}\label{topelement}
\ell(x) =\begin{cases}\of{\ell_{1}(x_1), \ldots, \ell_{m}(x_m)}^{\mathsf{T}} &\text{if }x \in \bigtimes_{i=1}^{m}\dom\ell_{i}, \\
+\infty & \text{else},
	\end{cases}
\end{equation}
for $x\in\R^m$. Similarly, the expected loss vector corresponding to a random position $X\in L_m^\infty$ is
$\E\sqb{\ell(-X)} \coloneqq 
\of{\E\sqb{\ell_1(-X_1)},\ldots,\E\sqb{\ell_m(-X_m)}}^{\mathsf{T}}$
if
$\Pr\cb{-X_i \in \dom\ell_i}=1$ for each $i\in\cb{1,\ldots,m}$, and $\E\sqb{\ell(-X)}\coloneqq +\infty$ otherwise.

\item \textbf{Comparison rule: }Let $C \subseteq \R^m$ be a closed convex set such that $C+\R^m_+ \subseteq C$ and $0 \in \R^m$ is a boundary point of $C$. Expected loss vectors will be compared according to the relation $\leq_C$ on $\R^m$ defined by
\begin{equation}\label{compC}
x \leq_C z \;\Leftrightarrow\; z \in x+ C.
\end{equation}
As $\R^m_+\subseteq C$, the relation $\leq_C$ provides a definition for a ``smaller" expected loss vector by generalizing the componentwise comparison of expected loss vectors with $\leq_{\R^m_+}$. Some examples of the set $C$ are discussed in Example~\ref{Cexamples} below.
\end{enumerate}
Finally, the incomplete preference relation $\succeq$ of the investor on $L_m^\infty$ is assumed to have the following numerical representation: For $X, Z\in L_m^\infty$,
\begin{equation}\label{exploss}
X\succeq Z \; \Leftrightarrow\; \E\sqb{\ell(-X)}\leq_C \E\sqb{\ell(-Z)}.
\end{equation}

\begin{rem}
In \eqref{compC} and \eqref{exploss}, the element $+\infty$ is added to $\R^m$ as a top element with respect to $\leq_C$, that is, $z \leq_C +\infty$ for every $z\in\R^m\cup\cb{+\infty}$. The addition on $\R^m$ is extended to $\R^m\cup\{+\infty\}$ by $z + (+\infty) = (+\infty) + z = +\infty$ for every $z \in \R^m\cup\{+\infty\}$.
\end{rem}

\begin{rem}
Note that $\leq_C$ (and hence $\succeq$) is reflexive (since $0 \in C$), transitive if $C + C \subseteq C$ and antisymmetric if $C \cap (-C) = \cb{0}$ ($C$ is ``pointed."). In particular, if $C$ is a pointed convex cone, then $\leq_C$ is a partial order which is compatible with the linear structure on $\R^m$.
\end{rem}

\begin{rem}
It is easy to check that $\succeq$ respects the complete preferences $\succeq_1, \ldots, \succeq_m$ on individual assets described in (i). In other words, for every $i\in\cb{1,\ldots,m}$, $X\in L_m^\infty$, and $Z_i\in L^\infty$,
\begin{equation}
X_i \succeq_i Z_i \;\Rightarrow\; X\succeq (X_1, \ldots, X_{i-1}, Z_i, X_{i+1}, \ldots, X_m).
\end{equation}
This is thanks to $\R^m_+\subseteq C$.
\end{rem}

\begin{rem}\label{motivation}
The choice of the componentwise structure of the vector loss function in \eqref{topelement} is justified by the following reasons.

\begin{enumerate}[(i)]
\item For each asset $i$, one could consider a more general loss function $\ell_i$ that depends on the vector $x \in \R^m$ but not only on the component $x_i$. However, the interconnectedness of the components of a portfolio $x=X(\omega)$ at time $t$ will be modeled in Section~\ref{market} by the prevailing exchange rates $\mathcal{C}_t(\omega)$ and trading constraints $\mathcal{D}_t(\omega)$, and thus, will be included in the market risk measure.

\item One could also consider vector loss functions $\ell \colon \R^d\rightarrow\R^m$ with $d>m$. The dimension reduction, which is motivated by allowing only $m$ of the $d$ assets to be used as eligible assets for risk compensation, is modeled by forcing liquidation into $L_{m}^{\infty}$ in Definition~\ref{marketextension} of the market risk measure. This includes the case where a large number of assets $d$ are denoted in a few ($m<d$) currencies, the currencies are used as eligible assets, and the loss functions are just defined for each of the $m$ currencies (but not for each asset individually).

\item On the other hand, the use of a vector loss function in the present paper is already more general than working under the assumption that there is a complete risk preference for multivariate positions (as in \citealt{BR06}) which even has a von Neumann - Morgenstern representation given by a real-valued loss function on $\R^m$ as in \cite{CampiOwen11} (see Example~\ref{Cexamples}(ii) below).
\end{enumerate}
\end{rem}

\begin{example}\label{Cexamples}
We consider the following examples of comparison rules for different choices of $C$.
\begin{enumerate}[(i)]
\item If $C=\R^{m}_{+}$, then $\leq_C=\leq$ corresponds to the componentwise ordering of the expected loss vectors. In this case, we simply have $\succeq = \succeq_1\times\ldots\times\succeq_m$.

\item If $C$ is a halfspace of the form $C=\cb{x\in\R^m\mid w^\mathsf{T} x\geq 0}$ for some $w\in\R^m_+\setminus\cb{0}$, then
\begin{equation}
X\succeq Z\; \Leftrightarrow\; \E\sqb{L(-X)}\leq \E\sqb{L(-Z)},
\end{equation}
where $x\mapsto L(x)\coloneqq \sum_{i=1}^m w_i\ell_{i}(x_i)$ is a multivariate real-valued loss function as in \citet[Example~2.10]{CampiOwen11}. In this case, $\succeq$ is a complete preference relation.

\item If $C$ is a polyhedral convex set of the form $C=\cb{x\in\R^m \mid Ax\geq b}$ for some $A\in\R^{n\times m}_+$, $b\in\R^n$, $n\geq 1$ (with $b_j = 0$ for some $j\in\cb{1,\ldots,n}$), then
\begin{equation}
X\succeq Z\; \Leftrightarrow\; \E\sqb{A(\ell(-X)-\ell(-Z))}\leq_{\R^n_+} b,
\end{equation}
which is a system of linear inequalities.
\end{enumerate}
\end{example}

\subsection{The shortfall risk measure and its set optimization formulation}\label{shortfall-sec}

Based on the incomplete preference relation $\succeq$ described in Section~\ref{incomplete-sec}, we define the shortfall risk measure next. To that end, for each $i\in\cb{1,\ldots,m}$, let $z^0_i \in \R$ such that $x^0_i\coloneqq \ell_i(-z^0_i)\in \interior\ell_i(\R)$. The point $z^0=(z^0_1,\ldots,z^0_m)^\mathsf{T}$ will be used as a \emph{deterministic benchmark} for multivariate random positions and  $x^0=(x^0_1,\ldots,x^0_m)^\mathsf{T}$ is the corresponding \emph{threshold value} for expected losses. Throughout, we assume that $x^0 = 0$. This is without loss of generality as otherwise one can shift the loss function and work with $x\mapsto \tilde{\ell}(x)=\ell(x)-x^0$. Recalling \eqref{compC} and \eqref{exploss}, note that
\begin{equation}
X\succeq z^0 \;\Leftrightarrow\; \E\sqb{\ell(-X)}\leq_C 0 \;\Leftrightarrow\; 0\in \E\sqb{\ell(-X)}+C\;\Leftrightarrow\; \E\sqb{\ell(-X)}\in - C. 
\end{equation}

The shortfall risk of a multivariate position $X\in L_m^\infty$ is defined as the set of all deterministic portfolio vectors $z\in\R^m$ that make $X+z$ preferable to the benchmark $z^0$.

\begin{defn}\label{set-srm-def}
The function $R_{\ell} \colon L_{m}^{\infty}\to\P(\R^{m})$ defined by
\begin{equation}
R_{\ell}(X)=\cb{z\in\R^m\mid X+z\succeq z^0 }
=\cb{z\in\R^{m}\mid\E\sqb{\ell(-X-z)}\in -C}
\end{equation}
is called the \emph{shortfall risk measure} (on $L_{m}^{\infty}$ with comparison rule $C$).
\end{defn}

In other words, the shortfall risk of $X\in L_m^\infty$ is the set of all vectors $z\in\R^m$ for which $X+z$ has a ``small enough" expected loss vector.

\begin{prop}\label{set-srm-lsc}
The shortfall risk measure $R_\ell$ satisfies the following properties:
\begin{enumerate}[(i)]
\item \textbf{Monotonicity:} $Z\geq X$ implies $R_\ell(Z)\supseteq R_\ell(X)$ for every $X, Z\in L_{m}^{\infty}$.
\item \textbf{Translativity:} $R_{\ell}(X+z)=R_{\ell}(X)-z$ for every $X\in L_{m}^{\infty}$, $z\in \R^{m}$.
\item \textbf{Finiteness at $0$:} $R_{\ell}(0)\notin\{\emptyset,\R^{m}\}$.
\item \textbf{Convexity:} $R_{\ell}(\lambda X+(1-\lambda)Z)\supseteq \lambda R_{\ell}(X)+(1-\lambda)R_{\ell}(Z)$ for every $X, Z\in L_{m}^{\infty}$, $\lambda\in(0,1)$.
\item \textbf{Weak$^\ast$-closedness:} The set $\graph R_\ell\coloneqq\{(X,z)\in L_{m}^{\infty}\times \R^{m}\mid z\in R_\ell(X)\}$ is closed with respect to the product of the $\text{weak}^{\ast}$ topology $\sigma(L_{m}^{\infty},L_{m}^{1})$ and the usual topology on $\R^{m}$.
\end{enumerate}
\end{prop}

\begin{rem}\label{interpret-rm}
The properties in Proposition~\ref{set-srm-lsc} make $R_\ell$ a sensible measure of risk for multivariate positions in the sense that every portfolio $z\in R_\ell(X)$ can compensate for the risk of $X\in L_m^\infty$. Let us comment on the financial interpretations of these properties. Monotonicity ensures that a larger position (with respect to componentwise ordering) is \emph{less} risky, that is, it has a \emph{larger} set of risk compensating portfolios. Translativity is the requirement that a deterministic increment to a position reduces each of its risk compensating portfolios by the same amount. Finiteness at $0$ guarantees that the risk of the zero position can be compensated by at least one but not all eligible portfolios. With the former two properties, it even guarantees \emph{finiteness everywhere} in the sense that $R_\ell(X)\notin\{\emptyset,\R^{m}\}$ for every $X\in L_{m}^{\infty}$. Convexity can be interpreted as the reduction of risk by diversification. Finally, weak$^\ast$-closedness is the set-valued version of the weak$^\ast$-lower semicontinuity property (as in \citealt{fs:sf}) for scalar risk measures.
\end{rem}

Set-valued functions satisfying the properties in Proposition~\ref{set-srm-lsc} are called \emph{(set-valued) \mbox{(weak$^\ast$-)} closed convex risk measures} and are studied in \cite{hh:duality}, \cite{hhr:setval}. An immediate consequence of these pro\-perties is that the values of a closed convex risk measure belong to the collection 
\begin{equation}\label{conveximage}
\G_m\coloneqq \G(\R^{m},\R^{m}_{+})\coloneqq\{A\subseteq \R^{m}\mid A=\cl\co(A+\R^{m}_{+})\},
\end{equation}
where $\cl$ and $\co$ denote the closure and convex hull operators, respectively. It turns out that $\G_m$ is a convenient image space\footnote{The phrase ``image space" for a set-valued function refers to the set (subset of a power set) where the function maps into. This set is not a linear space in general. In particular, $\G_m$ is a conlinear space in the sense of \cite{andreasduality}: It is closed under the closure of the Minkowski addition, and it is closed under multiplication by nonnegative scalars (with the convention $0\emptyset=\R^{m}_{+}$).} to study set optimization, see \cite{andreasduality}. In particular, it is an order complete lattice when equipped with the usual superset relation $\supseteq$. We have the following infimum and supremum formulae for every nonempty subset $\mathcal{A}$ of $\G_m$:
\begin{equation}\label{infsup}
\inf_{(\G_m,\supseteq)}\mathcal{A}=\cl\co\bigcup_{A\in\mathcal{A}}A,\quad \sup_{(\G_m,\supseteq)}\mathcal{A}=\bigcap_{A\in\mathcal{A}}A.
\end{equation}
The infimum formula is motivated by the fact that the union of closed (convex) sets is not necessarily closed (convex). We also use the conventions $\inf_{(\G_m,\supseteq)}\emptyset = \emptyset$ and $\sup_{(\G_m,\supseteq)}\emptyset = \R^{m}$.

Note that $C\in\G_m$ with $0$ being a boundary point of it. If $C=\R^{m}_{+}$, then the shortfall risk measure $R_{\ell}$ becomes a trivial generalization of the scalar shortfall risk measure (see Definition~\ref{srm-def}) in the sense that
\begin{equation}
R_{\ell}(X)=(\rho_{\ell_{1}}(X_{1}), \ldots,\rho_{\ell_{m}}(X_{m}))^{\mathsf{T}}+\R^{m}_{+},
\end{equation}
for every $X\in L_{m}^{\infty}$. In general, such an explicit representation of $R_{\ell}$ may not exist. However, given $X \in L_{m}^{\infty}$, one may write
\begin{equation}\label{set-primal-shortfall}
R_{\ell}(X)=\inf_{(\G_m,\supseteq)}\cb{z+\R_{+}^{m}\mid 0\in \mathbb{E}\sqb{\ell(-X-z)}+C,\; z\in\R^{m}},
\end{equation}
that is, $R_{\ell}(X)$ is the optimal value of the set minimization problem
\begin{equation}\label{setoptimization}
\text{minimize} \quad \Phi(z) \quad \text{subject to} \quad 0\in \Psi(z),\; z\in\R^{m},
\end{equation}
where $\Phi \colon \R^{m} \to \G_m$ and $\Psi \colon \R^{m} \to \G_m$ are the (set-valued) objective function and constraint function, respectively, defined by
\begin{equation}
\Phi(z)=z+\R^{m}_{+}, \quad \quad \Psi(z) = \mathbb{E}\sqb{\ell(-X-z)}+C.
\end{equation}
Here, it is understood that $\Psi(z) = \emptyset$ whenever $\E\sqb{\ell(-X-z)} = +\infty$.

A Lagrange duality theory for problems of the form \eqref{setoptimization} is developed in \cite{hl:lagrange}. Using this theory, we will compute the Lagrangian dual problem for $R_{\ell}(X)$. It will turn out that, after a change of variables, the dual objective function gives rise to another class of set-valued convex risk measures, called \emph{divergence risk measures}. We introduce these risk measures separately in Section~\ref{divergence-sec} first, and the duality results are deferred to Section~\ref{lagrange-sec}.

\subsection{Divergence risk measures}\label{divergence-sec}

In this section, we introduce divergence risk measures as decision-making problems of the investor about the level of consumption of a multivariate random position. The relationship between shortfall and divergence risk measures will be formulated in Section~\ref{lagrange-sec}.
 
Suppose that the investor with random portfolio $X\in L_m^\infty$ wants to choose a deterministic portfolio $z \in \R^m$ to be received at initial time. Hence, she will hold $X-z$ at terminal time. She has two competing objectives:
\begin{enumerate}[(i)]
\item \textbf{Maximizing consumption: }The investor wants to maximize her immediate consumption $z$, or equivalently, minimize $-z$. The optimization is simply with respect to the componentwise ordering of portfolio vectors. 

\item \textbf{Minimizing loss: }The investor wants to minimize the expected loss $\E\sqb{\ell(-X+z)}$ of the remaining random position $X-z$. In this case, the expected loss vectors are compared with respect to the set $C$ as in (iii) of Section~\ref{incomplete-sec}.
\end{enumerate}
These two objectives can be summarized as the following set minimization problem where the objective function maps into $\G_{2m}$: 
\begin{equation}
\text{minimize}\quad \begin{pmatrix} -z+\R^{m}_{+} \\ \mathbb{E}\sqb{\ell(-X+z)}+C\end{pmatrix} \quad \text{subject to}\quad  z\in\R^{m}.
\end{equation}
Here and in Definition~\ref{div-defn} below, the value of the objective function is understood to be $\emptyset$ if $\E\sqb{\ell(-X-z)}\notin\R^m$. On the other hand, the investor combines these competing objectives by means of a \emph{relative weight vector} $r\in\R^m_{+}$: For each asset $i\in\cb{1,\ldots,m}$, $r_i$ is the relative weight of the expected loss $\E\sqb{\ell_i(-X_i+z_i)}$ with respect to $-z_i$. As a result, she solves the ``partially scalarized" problem
\begin{equation}
\text{minimize} \quad -z+\R^{m}_{+}+r\cdot\of{\mathbb{E}\sqb{\ell(-X+z)}+C} \quad  \text{subject to} \quad z\in\R^{m}.
\end{equation}
The optimal value of this problem is defined as the divergence risk of $X$. 

\begin{defn}\label{div-defn}
Let $r\in\R^m_{+}$. The function $D_{\ell,r} \colon L_{m}^{\infty}\to\G_m$ defined by
\begin{equation}
D_{\ell,r}(X)=\inf_{(\G_m,\supseteq)}\cb{-z+r\cdot(\mathbb{E}\sqb{\ell(-X+z)}+C)\mid z\in\R^{m}},
\end{equation}
is called the \emph{divergence risk measure} with relative weight vector $r$.
\end{defn}

Apparently, for some values of $r\in\R^m_{+}$, the optimization problem is unbounded and one has $D_{\ell,r}(X)= \R^m$ for every $X\in L_m^\infty$. In Proposition~\ref{setvaluedoce}, we will characterize the set of all values of $r$ for which $D_{\ell,r}$ has finite values and indeed is a closed convex risk measure, that is, it satisfies the five properties in Proposition~\ref{set-srm-lsc}.

\begin{rem}
In the one-dimensional case $m=1$,  one has $D_{\ell,r}(X)=\delta_{\ell,r}(X)+\R_+$, where $\delta_{\ell,r}(X)=\inf_{z\in\R\colon\E\sqb{\ell(-X+z)}<+\infty}\of{-z+r\E\sqb{\ell(-X+z)}}$ is the divergence risk measure as in Definition~\ref{divergencermdefn}. In the literature (see \citealt{oldoce,bt:oce}), only the case $r=1$ is considered in the definition of divergence risk measure (or \emph{optimized certainty equivalent} for $-\delta_{\ell,1}$). In \cite{bt:oce}, the general case $r>0$ is simply treated with a scaled loss function $r\ell$ since $\delta_{\ell,r}=\delta_{r\ell,1}$. (This simplification is not possible in the multidimensional case $m>1$ as $r\in\R^m_{+}$ is multiplied by the set $\E\sqb{\ell(-X+z)}+C$ but not only the vector $\E\sqb{\ell(-X+z)}$.)

However, in our treatment, $\delta_{\ell,r}$ is interpreted as a weighted sum scalarization of a bi-objective optimization problem and this problem is, in turn, characterized by the whole family $\of{\delta_{\ell,r}}_{r\in\R_{+}}$ of divergence risk measures. This interpretation is an original contribution of the present paper to the best of our knowledge.
\end{rem}

\subsection{The Lagrange dual formulation of the shortfall risk measure}\label{lagrange-sec}

This section formulates one of the main results of the paper, Theorem~\ref{srmdrmconnection}: The shortfall risk measure can be written as the intersection, that is, the set-valued supremum (see \eqref{infsup}), of divergence risk measures indexed by their relative weight vectors. 

The result is derived in Section~\ref{setvaluedproof} using the recent Lagrange duality in \cite{hl:lagrange} applied to the shortfall risk measure as the primal problem. The Lagrange duality is reviewed in Section~\ref{lagrange-app}. The result of its direct application to the shortfall risk measure is stated in Lemma~\ref{set-srm-simple-dual}, followed by a change of variables provided in Lemma~\ref{changeofvariable}. This additional latter step is essential in obtaining divergence risk measures in the (reformulated) dual problem.

Recall that the conjugate function of $\ell_i$ is denoted by $g_i$, which is a divergence function in the sense of Definition~\ref{divergence}. The \emph{vector divergence function} $g \colon \R^m \to \R^m\cup\{+\infty\}$ is defined by
\begin{equation}\label{vectordomain}
g(y) =\begin{cases}\of{g_{1}(y_1), \ldots, g_{m}(y_m)}^{\mathsf{T}} &\text{if }y \in \dom g\coloneqq\bigtimes_{i=1}^m\dom g_{i},\\
+\infty & \text{else}.
	\end{cases}
\end{equation}
In view of Remark~\ref{divergenceindex}, given $r\in\R^m_+$, we define 
\begin{equation}
g_r(y)=((g_1)_{r_1}(y_1),\ldots,(g_m)_{r_m}(y_m))^{\mathsf{T}}
\end{equation}
for each $y\in\R^m$ and set $\dom g_r = \bigtimes_{i=1}^m \dom (g_{i})_{r_i}$. Moreover, for $r\in\R^m_{++}$, we write $\frac{1}{r}\coloneqq(\frac{1}{r_1},\ldots,\frac{1}{r_m})^{\mathsf{T}}$.

\begin{thm}\label{srmdrmconnection}
For every $X\in L_{m}^{\infty}$,
\begin{equation}\label{shortfallmaxdivergence1}
R_{\ell}(X)=\sup_{(\G_m,\supseteq)}\cb{D_{\ell,r}(X)\mid r\in \R^m_{+},1\in\dom g_r}
=\bigcap_{r\in \R^m_{+}\colon 1\in \dom g_r}D_{\ell,r}(X).
\end{equation}
In particular, if $\dom\ell\coloneqq\bigtimes_{i=1}^m\dom\ell_i=\R^m$, then
\begin{align}\label{shortfallmaxdivergence}
R_{\ell}(X)=\sup_{(\G_m,\supseteq)}\cb{D_{\ell,r}(X)\mid r\in \R^m_{++}, \frac{1}{r}\in \dom g}=\bigcap_{r\in \R^m_{++}\colon \frac{1}{r}\in \dom g}D_{\ell,r}(X).
\end{align}
\end{thm}

\begin{rem}\label{solconcept}
Theorem~\ref{srmdrmconnection} shows that the shortfall risk measure can be computed as a set-valued supremum over divergence risk measures. However, in general, there is no single $r \in\R^m_{+}$ with $ 1\in\dom g_r$ which yields this supremum, that is, the supremum is not attained at a single argument. Instead, one could look for a set $\Gamma \subseteq\R^m_{+}$ with $1\in\dom g_r$ for every $r\in\Gamma$ that satisfies the following two conditions: \eqref{shortfallmaxdivergence1} holds with the intersection running through all $r\in\Gamma$, and each $D_{\ell,r}(X)$ with $r \in \Gamma$ is a maximal element of the set $\cb{D_{\ell,r}(X) \mid r\in \R^m_{+}, 1\in \dom g_r}$ with respect to $\supseteq$. This corresponds to the  solution concept for set optimization problems introduced in \cite{heydeloehne11} (see also \citealt[Definition~3.3]{hl:lagrange}) and will be discussed for the entropic risk measure in Section~\ref{entropicexample} together with the precise definition of a maximal element.
\end{rem}

For every $r \in \R^{m}_{+}$, define a function $\delta_{\ell,r} \colon L_{m}^{\infty} \to \R^m\cup\cb{-\infty}$ by 
\begin{equation}\label{vectordrm}
\delta_{\ell,r}(X) = (\delta_{\ell_{1},r_{1}}(X_{1}), \ldots, \delta_{\ell_{m},r_{m}}(X_{m}))^{\mathsf{T}}
\end{equation}
whenever the right hand side is in $\R^m$ and $\delta_{\ell,r}(X) = -\infty$ otherwise. Recall that $\delta_{\ell_{i},r_{i}}$ is given by
\begin{equation}\label{EqDualDiv1}
\delta_{\ell_{i},r_{i}}(X_{i})=\inf_{z_{i}\in\R}\of{z_{i}+r_{i}\E\sqb{\ell_{i}(-X_{i}-z_{i})}}
\end{equation}
if $r_i >0$, and we have $\delta_{\ell_{i},0}(X_i)=-\essinf X_i -\sup\dom\ell_i$; see \eqref{divergencerisk}. If $1\in\dom (g_i)_{r_i}$, then $\delta_{\ell_i,r_i}$ is the scalar $(\ell_i,r_i)$-divergence risk measure according to Definition~\ref{divergencermdefn}.

As a byproduct of Theorem~\ref{srmdrmconnection}, we show that a divergence risk measure has a much simpler form in terms of scalar divergence risk measures. 

\begin{prop}\label{setvaluedoce}
Let $r\in \R^m_+$.
\begin{enumerate}[(i)]
\item If $1\in \dom g_r$, then $D_{\ell, r}$ is a closed convex risk measure with the representation
\begin{equation}
D_{\ell,r}(X) =\delta_{\ell,r}(X)+r\cdot C.
\end{equation}
\item Otherwise, $D_{\ell,r}(X)=\R^m$ for every $X\in L^\infty_m$.
\end{enumerate}
In particular, if $\dom\ell=\R^m$, then $D_{\ell,r}$ is a closed convex risk measure if and only if $r\in\R^{m}_{++}$ with $\frac{1}{r}\in\dom g$.
\end{prop}

Note that, in the representation in Proposition~\ref{setvaluedoce}, the dependence on $X\in L_m^\infty$ is only through the vector part $\delta_{\ell,r}(X)$; however, the choice of the relative weight vector $r$ still affects the distortion on the set $C$ through $r\cdot C$.

\begin{rem}
Let us comment on the trade-off between using a shortfall risk measure and a divergence risk measure. According to the representation in Proposition~\ref{setvaluedoce}(i), the divergence risk measure with relative weight vector $r\in \R^m_{+}$ (with $1\in\dom g_r$) has the simple structure $D_{\ell,r}(X)=\delta_{\ell,r}(X)+r\cdot C$, where the dependence on $X$ is solely on the vector $\delta_{\ell,r}(X)$ of scalar divergence risk measures. Hence, the computation of the divergence risk measure reduces to the computation of $m$ scalar risk measures. However, the shortfall risk measure does not possess such a simple representation as a constrained optimization problem. It is a (much) more conservative notion of set-valued risk as $R_{\ell}(X)\subseteq D_{\ell,r}(X)$. On the other hand, the divergence risk measure has the additional parameter $r$: For each asset, the investor has to specify how many units of her expected loss is comparable with one unit of the consumption at initial time.
\end{rem}

\subsection{Dual representations}

In this section, we state representations of shortfall and divergence risk measures in terms of \emph{vector probability measures} and \emph{weight vectors}. Such representations of convex risk measures are called \emph{dual representations}.

In \cite{hh:duality} and \cite{hhr:setval}, it is shown that a closed convex risk measure is indeed characterized by a halfspace-valued function that shows up in its dual representation. We recall this result first. To that end, let $\Q=(\Q_{1},\ldots,\Q_{m})^{\mathsf{T}}$ be an \emph{$m$-dimensional vector probability measure} in the sense that $\Q_{i}$ is a probability measure on $(\Omega,\mathcal{F})$ for each $i\in\{1,\ldots,m\}$. We define $\E^{\Q}\sqb{X}=(\E^{\Q_{1}}\sqb{X_{1}},\ldots,\E^{\Q_{m}}\sqb{X_{m}})^{\mathsf{T}}$ for every $X\in L_{m}^{0}$ such that the components exist in $\R$. We denote by $\mathcal{M}_{m}(\Pr)$ the set of all $m$-dimensional vector probability measures on $(\Omega,\mathcal{F})$ whose components are absolutely continuous with respect to $\Pr$. For $\Q\in\mathcal{M}_{m}(\Pr)$, we set
$\frac{d\Q}{d\Pr}=(\frac{d\Q_{1}}{d\Pr}, \ldots,\frac{d\Q_{m}}{d\Pr})^{\mathsf{T}}$, where, for each $i\in\{1,\ldots,m\}$, $\frac{d\Q_{i}}{d\Pr}$ denotes the Radon-Nikodym derivative of $\Q_{i}$ with respect to $\Pr$. For $w\in\R^{m}_{+}\bs\{0\}$, we define the halfspace
\begin{equation}
G(w) \coloneqq \cb{z \in\R^m \mid w^{\mathsf{T}}z \geq 0}.
\end{equation}

\begin{prop}\label{dualthm}
(\citealt[Theorem 4.2]{hhr:setval}) A function $R \colon L_{m}^{\infty}\to\G_m$ is a closed convex risk measure if and only if, for every $X \in L_{m}^{\infty}$,
\begin{equation}\label{dualrep}
R(X)=\bigcap_{(\Q,w)\in\M_{m}(\Pr)\times (\R^{m}_{+}\bs\{0\})} \of{-\alpha_{R}(\Q,w)+\E^{\Q}\sqb{-X}},
\end{equation}
where $-\alpha_{R}:\M_{m}(\Pr)\times (\R^{m}_{+}\bs \{0\})\to\G_m$ is the penalty function of $R$ defined by
\begin{equation}\label{minpen}
-\alpha_{R}(\Q,w)=\cl\bigcup_{X\in L_{m}^{\infty}}\sqb{R(X)+\of{\E^{\Q}\sqb{X}+G(w)}},
\end{equation}
for each $(\Q,w)\in\M_{m}(\Pr)\times (\R^{m}_{+}\bs\{0\})$.
\end{prop}

As in the scalar case, the penalty function of a closed convex risk measure basically coincides with its Fenchel conjugate. In the set-valued case, the transformation from the set-valued conjugate with dual variables $L^{1}_{m} \times (\R^m_+\bs\{0\})$ to a penalty function with dual variables in $\M_m(\Pr) \times (\R^m_+\bs\{0\}) $ requires some extra care; this procedure is described in detail in \cite{hh:duality}, \cite{hhr:setval}.

In Proposition~\ref{set-drm-penalty} and Proposition~\ref{set-srm-penaltyname}, we present the penalty functions of divergence and shortfall risk measures, respectively. To that end, we define a divergence for a vector probability measure.

\begin{defn}\label{vectordiv}
Let $r \in\R^m_+$ and $\Q\in\M_{m}(\Pr)$. For each $i\in\{1,\ldots,m\}$, let
\begin{equation}
I_{g_{i},r_{i}}(\Q_{i}\mid\Pr)\coloneqq\begin{cases} r_{i}\E\sqb{g_{i}\of{\frac{1}{r_{i}}\frac{d\Q_{i}}{d\Pr}}}&\text{if }r_i>0,\\ \sup\dom\ell_i &\text{if }r_i=0.\end{cases}
\end{equation}
The element $I_{g,r}(\Q\mid\Pr) \in \R^m\cup\{+\infty\}$ defined by
\begin{equation}
I_{g,r}(\Q\mid\Pr)\coloneqq(I_{g_{1},r_{1}}(\Q_{1}\mid\Pr),\ldots, I_{g_{m},r_{m}}(\Q_{m}\mid\Pr))^{\mathsf{T}}
\end{equation}
if $I_{g_{i},r_{i}}(\Q_{1}\mid\Pr)\in\R$ for each $i\in\cb{1,\ldots,m}$, and by $I_{g,r}(\Q\mid\Pr)\coloneqq+\infty$ otherwise is called the \emph{vector $(g,r)$-divergence} of $\Q$ with respect to $\Pr$.
\end{defn}

Note that $I_{g_{i},r_{i}}(\Q_i \mid \Pr)$ is the (scalar) $(g_{i},r_{i})$-divergence of $\Q_{i}$ with respect to $\Pr$, see Definition~\ref{divergencedefn}.

\begin{prop}\label{set-drm-penalty}
Let $r\in \R^m_{+}$ with $1\in \dom g_r$. The penalty function of the divergence risk measure $D_{\ell,r}$ is given by
\begin{equation}\label{setdiver}
-\alpha_{D_{\ell,r}}(\Q,w) = -I_{g,r}(\Q\mid\Pr)+r\cdot C+G(w)
\end{equation}
for each $(\Q,w)\in\M_{m}(\Pr)\times(\R^{m}_{+}\bs\{0\})$ with the convention $-\alpha_{D_{\ell,r}}(\Q,w)=\R^m$ if $I_{g,r}(\Q\mid\Pr)=+\infty$.
\end{prop}

\begin{prop}\label{set-srm-penaltyname}
The penalty function of the shortfall risk measure $R_{\ell}$ is given by
\begin{align}\label{set-srm-penalty}
-\alpha_{R_{\ell}}(\mathbb{Q},w)
&=\left\{z\in\R^{m}\mid w^{\mathsf{T}}z\geq  \sup_{r\in\R^{m}_{+}}\left(-w^{\mathsf{T}}I_{g,r}(\Q\mid\Pr)+\inf_{x\in C}{w^{\mathsf{T}}(r\cdot x)}\right)\right\}\notag  \\
&=\bigcap_{r\in \R^m_{+}\colon 1\in \dom g_r}-\alpha_{D_{\ell,r}}(\Q,w).
\end{align}
for each $(\Q,w)\in\M_{m}(\Pr)\times \R^{m}_{++}$ with the convention $-\alpha_{D_{\ell,r}}(\Q,w)=\R^m$ if \mbox{$I_{g,r}(\Q\mid\Pr)=+\infty$}. In particular, if $\dom \ell =\R^m$, then
\begin{equation}
-\alpha_{R_{\ell}}(\mathbb{Q},w)=\bigcap_{r\in \R^m_{++}\colon \frac{1}{r}\in \dom g}-\alpha_{D_{\ell,r}}(\Q,w).
\end{equation}
\end{prop}

\section{Examples}\label{examples}
\subsection{Set-valued entropic risk measures}\label{entropicexample}

In this section, we assume that the vector loss function $\ell$ of Section~\ref{mainsection} is the \emph{vector exponential loss function} with constant risk aversion vector $\beta\in\R^{m}_{++}$, that is, for each $i\in\{1,\ldots,m\}$ and $x\in\R$,
\begin{equation}
\ell_{i}(x)=\frac{e^{\beta_{i}x}-1}{\beta_{i}},
\end{equation}
which satisfies the conditions in Definition~\ref{loss}. The corresponding vector divergence function $g$ is given by
\begin{equation}
g_{i}(y)=\frac{y}{\beta_{i}}\log y-\frac{y}{\beta_{i}}+\frac{1}{\beta_{i}},
\end{equation}
for each $i\in\{1,\ldots,m\}$ and $y\in\R$. Here and elsewhere, we make the convention $\log y=-\infty$ for every $y\leq 0$.

For convenience, we sometimes use the notation $[x_{i}]_{i=1}^{m}$ for $x=(x_{1},\ldots,x_{m})^{\mathsf{T}}\in\R^{m}$. Let us also define $x^{-1}\coloneqq (x_1^{-1},\ldots,x_m^{-1})^{\mathsf{T}}$ and $\log x \coloneqq (\log x_1, \ldots, \log x_m)^{\mathsf{T}}$ for $x\in\R^m_{++}$, and $\log[A]\coloneqq\{\log x\mid x\in A\}$ for $A \subseteq\R^{m}_{++}$. We will also use $1\coloneqq(1,\ldots,1)^{\mathsf{T}}$ as an element of $\R^m$.

Note that $\interior\ell(\dom\ell)=\ell(\dom\ell)=\ell(\R)=-\beta^{-1}+\R^m_{++}$ so that $0\in\interior\ell(\dom\ell)$. Let $C\in\G_m$ with $0\in\R^m$ being a boundary point of $C$. We call the corresponding shortfall risk measure $R^{\ent}\coloneqq R_{\ell}$ the \emph{entropic risk measure}. The next proposition shows that $R^{\ent}$ has the simple form of ``a vector-valued function plus a fixed set,'' which is, in general, not the case for an arbitrary loss function. Note that the functional $\rho^{\ent}$ in Proposition~\ref{pointplusset} is the vector of scalar entropic risk measures.

\begin{prop}\label{pointplusset}
For every $X\in L_{m}^{\infty}$,
\begin{equation}
R^{\ent}(X) = \rho^{\ent}(X)+C^{\ent},
\end{equation}
where
\begin{align}\label{componentwiseentropic}
&\rho^{\ent}(X) \coloneqq \sqb{\frac{1}{\beta_{i}}\log\E\sqb{e^{-\beta_{i}X_{i}}}}_{i=1}^m,\\
&C^{\ent} \coloneqq-\beta^{-1}\cdot \log\sqb{(1-\beta\cdot  C)\cap \R^{m}_{++}}.\notag 
\end{align}
\end{prop}

Note that the set $\dom g$ defined in \eqref{vectordomain} becomes $\R^{m}_{+}$. Since $\dom \ell = \R^m$, by Proposition~\ref{setvaluedoce}, $D^{\ent}_{r}\coloneqq D_{\ell,r}$ is a closed convex risk measure (divergence risk measure) if $r\in\R^m_{++}$ and $D_{\ell,r}(X)=\R^m$ for every $X\in L_m^\infty$ if $r\in\R^m_+\bs\R^m_{++}$.

\begin{prop}\label{ent-drm}
For every $r \in\R^{m}_{++}$ and $X \in L_{m}^{\infty}$,
\begin{equation}
D^{\ent}_{r}(X)=\rho^{\ent}(X)+\beta^{-1}\cdot (1-r+\log r)+r\cdot C,
\end{equation}
where $\rho^{\ent}(X)$ is defined by \eqref{componentwiseentropic}.
\end{prop}

Recall from \eqref{shortfallmaxdivergence} that $R^{\ent}(\cdot)$ is the supremum of all $D^{\ent}_{r}(\cdot)$ with $r\in\R^{m}_{++}$, that is, for $X\in L_{m}^{\infty}$,
\begin{equation}\label{entropicmaximization}
 R^{\ent}(X)  =\sup_{r\in\R^{m}_{++}}D_{r}^{\ent}(X)=\bigcap_{r\in\R^{m}_{++}}D_{r}^{\ent}(X).
\end{equation}
If $m=1$, then the only choice for $C$ is $\R_{+}$. In this case, one can check that, for $X\in L^{\infty}$,
\begin{equation}
 R^{\ent}(X)=D_{1}^{\ent}(X) = \rho^{\ent}(X)+\R_+.
\end{equation}
In other words, the supremum in \eqref{entropicmaximization} is \emph{attained} at $r=1$. In general, when $m\geq 2$, we may not be able to find some $\bar{r}\in\R_{++}$ for which $R^{\ent}(X)=D^{\ent}_{\bar{r}}(X)$. Instead, we will compute a \emph{solution} to this set maximization problem in the sense of \citet[Definition~3.3]{hl:lagrange}, that is, we will find a set $\Gamma\subseteq \R^m_{++}$ such that
\begin{enumerate}[(i)]
\item $R^{\ent}(X)=\bigcap_{r\in\Gamma}D^{\ent}_{r}(X)$,
\item for each $\bar{r}\in\Gamma$, $D_{\bar{r}}^{\ent}(X)$ is a \emph{maximal element} of the collection $\{D_{r}^{\ent}(X)\mid r\in\R^{m}_{++}\}$ in the following sense:
\begin{equation}\label{maximal}
\forall r\in\R^{m}_{++}:\quad D^{\ent}_{r}(X)\subseteq D^{\ent}_{\bar{r}}(X)\;\Rightarrow\;r=\bar{r}.
\end{equation}
\end{enumerate}
Moreover, the set $\Gamma$ will be independent of the choice of $X$. To that end, by Proposition~\ref{ent-drm}, we can rewrite $D^{\ent}_{r}(X)$ as
\begin{equation}
D^{\ent}_{r}(X)=\rho^{\ent}(X)+\bigcap_{w\in\R^{m}_{+}\bs\{0\}}\{z\in\R^{m}\mid w^{\mathsf{T}}z\geq -(f_{w}(r)+h_{w}(r))\},
\end{equation}
where, for $w\in\R^{m}_{+}\bs\{0\}$, $r\in\R^{m}_{++}$,
\begin{align}
&f_{w}(r)\coloneqq w^{\mathsf{T}}\of{-\beta^{-1}\cdot (1-r+\log r) },\\ &h_{w}(r)\coloneqq -\inf_{x\in C}w^{\mathsf{T}}(r\cdot x)=\sup_{x\in -C}w^{\mathsf{T}}(r\cdot x).\notag 
\end{align}

\begin{lem}\label{bestratio}
Let $w\in\R^{m}_{+}\bs\{0\}$. The function $f_{w}+h_{w}$ on $\R^{m}_{++}$ is either identically $+\infty$ or else it attains its infimum at a unique point $r^{w}\in\R^{m}_{++}$ which is determined by the following property: $r^{w}$ is the only vector $r\in\R^{m}_{++}$ for which $C$ is supported at the point
$\beta^{-1}\cdot (1-r^{-1})$ by the hyperplane with normal direction $r\cdot w$.
\end{lem}

\begin{prop}\label{ent-setopt}
Using the notation in Lemma~\ref{bestratio}, the set
\begin{equation}
\Gamma\coloneqq\left\{r^{w}\mid w\in\R^{m}_{+}\bs\{0\},\; f_{w}+h_{w}\text{ is proper}\right\}
\end{equation}
is a solution to the maximization problem in \eqref{entropicmaximization} for every $X\in L_{m}^{\infty}$.
\end{prop}

Finally, we compute the penalty function of $R^{\ent}$ in terms of the \emph{vector relative entropies}
\begin{equation}\label{relent}
H(\Q\mid \Pr)\coloneqq \sqb{\E^{\Q_{i}}\sqb{\log\frac{d\Q_{i}}{d\Pr}}}_{i=1}^{m}
\end{equation}
of vector probability measures $\Q\in\M_{m}(\Pr)$. Thus, the penalty function for the entropic risk measure is of the form ``negative vector relative entropy plus a nonhomogeneous halfspace'' (except for the trivial case).

\begin{prop}\label{ent-dual}
For every $(\Q,w)\in\M_{m}(\Pr)\times(\R^{m}_{+}\bs\{0\})$, we have $-\alpha_{R^{\ent}}(\Q,w)=\R^{m}$ if $h_{w}(r)=+\infty$ for every $r\in\R^m_{++}$, and 
\begin{equation}
-\alpha_{R^{\ent}}(\Q,w)=-\beta^{-1}\cdot H(\Q \mid \Pr) - \beta^{-1} \cdot \log\sqb{(1 - \beta\cdot C) \cap \R^m_{++}} + G(w)
\end{equation}
if $h_{w}$ is a proper function.
\end{prop}

\subsection{Set-valued average values at risks}

In this section, we assume that the vector loss function $\ell$ of Section~\ref{mainsection} is the \emph{(vector) scaled positive part function} with scaling vector $\alpha\in (0,1]^{m}$, that is, for each $i\in\{1,\ldots,m\}$ and $x\in\R$,
\begin{equation}
\ell_{i}(x)=\frac{x^{+}}{\alpha_{i}},
\end{equation}
which satisfies the conditions in Definition~\ref{loss}. The corresponding vector divergence function $g$ is given by
\begin{equation}
g_{i}(y)=\begin{cases}0 & \text{if }y\in\left[0,\frac{1}{\alpha_{i}}\right],\\ +\infty &\text{else},\end{cases}
\end{equation}
for each $i\in\cb{1,\ldots,m}$ and $y\in\R$.

Note that $0\not\in\interior\ell(\dom\ell)=\R^{m}_{++}$ in this example. Hence, let us fix $x^0 \in \R^m_{++}$ and $C\in\G_m$ with $0$ being a boundary point of $C$. We will apply the definitions and results of Section~\ref{mainsection} to the shifted loss function $\tilde{\ell}(x)=\ell(x)-x^0$. The corresponding shortfall risk measure is given by
\begin{equation}
R_{\tilde{\ell}}(X)=\cb{z\in\R^{m}\mid \E\sqb{(z-X)^{+}}\in \alpha\cdot (x^{0}-C)},
\end{equation}
where the positive part function is applied componentwise.

Note that the set $\dom g$ defined in \eqref{vectordomain} becomes $\bigtimes_{i=1}^{m}[0,\frac{1}{\alpha_{i}}]$. Since $\dom \tilde{\ell} = \R^m$, by Proposition~\ref{setvaluedoce}, $D_{\tilde{\ell},r}$ is a closed convex risk measure (divergence risk measure) if $r\in\bigtimes_{i=1}^m [\alpha_i,+\infty)$ and $D_{\tilde{\ell},r}(X)=\R^m$ for every $X\in L_m^\infty$ if $r\in\R^m_+\bs\bigtimes_{i=1}^m [\alpha_i,+\infty)$. In the former case, the divergence risk measure with relative weight vector $r\in \bigtimes_{i=1}^m [\alpha_i,+\infty)$ is given by
$D_{\tilde{\ell},r}(X)=\delta_{\tilde{\ell},r}(X)+r\cdot C$
for $X\in L_{m}^{\infty}$, where, for each $i\in\{1,\ldots,m\}$,
\begin{equation}
 \delta_{\tilde{\ell}_{i},r_{i}}(X_{i})=\inf_{z_{i}\in\R}\of{-z_{i}+\frac{r_{i}}{\alpha_{i}}\E\sqb{(z_{i}-X_{i})^{+}}}-r_{i}x^{0}_{i}.
\end{equation}
When $r=(1,\ldots,1)^{\mathsf{T}}$ and $C=\R^{m}_{+}$, we obtain the set-valued average value at risk in the sense of \citet[Definition~2.1 for $M=\R^{m}$]{hry:avar}, which is given by
\begin{align}
\Avar_{\alpha}(X)\coloneqq &D_{\tilde{\ell},1}(X)+x^{0}\\
=&\sqb{\inf_{z_{i}\in\R}\of{-z_{i}+\frac{1}{\alpha_{i}}\E\sqb{(z_{i}-X_{i})^{+}}}}_{i=1}^{m}+\R^{m}_{+}.\notag 
\end{align}
Hence, our framework offers the following generalization of the set-valued average value at risk as a convex risk measure:
\begin{align}
\Avar_{\alpha,r}(X)\coloneqq & D_{\tilde{\ell},r}(X)+r\cdot x^{0}\\
=&\sqb{\inf_{z_{i}\in\R}\of{-z_{i}+\frac{r_{i}}{\alpha_{i}}\mathbb{E}\sqb{(z_{i}-X_{i})^{+}}}}_{i=1}^{m}+r\cdot C.\notag 
\end{align}
As in the scalar case, this definition even works for $X\in L_{m}^1$.

\section{Market risk measures}\label{market}

The purpose of this section is to propose a method to incorporate the frictions of the market into the quantification of risk. As the first step of the method, it is assumed that there is a ``pure" risk measure $R$ that represents the attitude of the investor towards the assets of the market. This could be one of the utility-based risk measures introduced in Section~\ref{mainsection}. Since the risk measure $R$ does not take into account the frictions of the market, the second step consists of minimizing risk subject to the trading opportunities of the market. More precisely, we minimize (in the sense of set optimization) the value of $R$ over the set of financial positions that can be reached with the given position by trading in the so-called \emph{convex market model}. The result of the risk minimization, as a function of the given position, is called the \emph{market risk measure} induced by $R$.

In the literature, minimization of scalar risk measures subject to trading constraints are considered in \cite{barrieu}. In the multivariate case, market risk measures are introduced in \cite{hhr:setval} and \cite{hry:avar} for the special case of a conical market model. Here, this notion is considered for an arbitrary convex risk measure with the more general convex market model of \cite{PP10} and the possibility of trading constraints and liquidation into fewer assets. The market is described in Section~\ref{marketmodel}. The dual representation result, Theorem~\ref{liquidationtheorem} in Section~\ref{marketmain}, is one of the main contributions of this paper. Finally, in Section~\ref{marketconditions}, we present sufficient conditions under which Theorem~\ref{liquidationtheorem} can be applied to shortfall and divergence risk measures. 

\subsection{The convex market model with trading constraints}\label{marketmodel}

Consider a financial market with $d\in\{1,2,\ldots\}$ assets. We assume that the market has convex transaction costs or nonlinear illiquidities in finite discrete time. Following \cite{PP10}, we use convex solvency regions to model such frictions. To that end, let $T\in\{1,2,\ldots\}$, $\T=\{0,\ldots,T\}$, and $(\F_{t})_{t\in\T}$ a filtration of $(\O,\F,\Pr)$ augmented by the $\Pr$-null sets of $\F$. The number $T$ denotes the time horizon, and $(\F_{t})_{t\in\T}$ represents the evolution of information over time. We suppose that there is no information at time $0$, that is, every $\F_{0}$-measurable function is deterministic $\Pr$-almost surely; and there is full information at time $T$, that is, $\F_T=\F$. For $p\in \cb{0,1,+\infty}$ and $t\in\T$, we denote by $L_{d}^{p}(\F_t)$ the linear subspace of all $\F_t$-measurable random variables in $L_{d}^{p}$.

Let $t\in\T$. By the $\F_t$-measurability of a set-valued function $D \colon \O\to \P(\R^{d})$, it is meant that the graph $\{(\o,y)\in\O\times\R^{d}\mid y\in D(\o)\}$ is $\F_{t}\otimes\B(\R^{d})$-measurable, where $\B(\R^{d})$ denotes the Borel $\sigma$-algebra on $\R^{d}$. For such function $D$, define the set
$L_{d}^{p}(\F_t,D)\coloneqq$$\{Y\in L_{d}^{p}(\F_t)\mid \Pr\{\o\in\O\mid Y(\o)\in D(\o)\}=1\}$ for $p\in\cb{0,1,+\infty}$.

We recall the convex market model of \cite{PP10} next. For each $t \in\T$, let $\C_{t} \colon \O \to\G_d$ be an $\F_{t}$-measurable function such that $\R^{d}_{+}\subseteq \C_{t}(\o)$ and \mbox{$-\R^d_+\cap\C_t(\o)=\{0\}$} for each $t\in\{0, \ldots,T\}$ and $\o\in\O$. The set $\C_{t}$ is called the \emph{(random) solvency region} at time $t$; see \cite{AT07}, \cite{PP10}. It models the bid and ask prices as a function of the magnitude of a trade, for instance, as in \cite{CJP04}, \cite{CR07}, \cite{RS10}; and thus, directly relates to the shape of the order book. More precisely, $\C_t(\o)$ is the set of all portfolios which can be exchanged into ones with nonnegative components at time $t$ when the outcome is $\o$. Convex solvency regions allow for the modeling of temporary illiquidity effects in the sense that they cover nonlinear illiquidities; however, they assume that agents have no market power, and thus, their trades do not affect the costs of subsequent trades.

\begin{example}\label{conical}
An important special case is the conical market model introduced in \cite{kabanov}. Suppose that $\C_t(\o)$ is a (closed convex) cone for each $t\in\T$ and $\o\in\O$. In this case, the transaction costs are proportional to the size of the orders. 
\end{example}

From a financial point of view, it is possible to have additional constraints on the trading opportunities at intermediate times. For instance, trading may be allowed only up to a (possibly state- and time-dependent) threshold level for the \mbox{assets} (\mbox{Example}~\ref{threshold}), or it may be the case that a certain linear combination of the \mbox{trading} units should not exceed a threshold level (Example~\ref{lincomb}). Such constraints are \mbox{modeled} via convex random sets. Given $t\in\{0,\ldots,T-1\}$, let $\D_{t}:\O\to \P(\R^{d})$ be an $\F_{t}$-measurable function such that $\D_{t}(\o)$ is a closed convex set and $0\in \C_{t}(\o)\cap \D_{t}(\o)$ for every $\o\in\O$. Note that $\D_{t}$ does not necessarily map into $\G_d$, and this is why we prefer to work with $\C_t\cap \D_t$ instead of replacing $\C_t$ by $\C_t \cap \D_t$. For convenience, let us also set $\D_T(\o)=\R^d$ for every $\o\in\O$.

\begin{example}\label{threshold}
For each $t\in\{0,\ldots,T-1\}$, suppose that $\D_{t}=\bar{Y}_t - \R^d_+$, for some $\bar{Y}_{t}\in L_{d}^{0}(\F_{t},\R^{d}_{+})$. In this case, trading in asset $i\in\{1,\ldots,d\}$ at time $t\in\{0,\ldots,T-1\}$ may not exceed the level $(\bar{Y}_{t})_{i}$.
\end{example}

\begin{example}\label{lincomb}
For each $t\in\{0,\ldots,T-1\}$, suppose that $\D_{t}=\{y\in\R^{d}\mid A_{t}^{\mathsf{T}}y\leq B_{t}\}$, for some $A_{t}\in L_{d}^{0}(\F_{t},\R^{d}_{+}\bs\{0\})$ and $B_{t}\in L_{1}^{0}(\mathcal{F}_{t},\R_{+})$. In this case, trading in each asset is unlimited but the linear combination of the trading units with the weight vector $A_{t}$ cannot exceed the level $B_{t}$.
\end{example}

The set of all financial positions that can be obtained by trading in the market starting with the zero position is
\begin{equation}\label{freelyavailable}
\K\coloneqq -\sum_{t=0}^{T}L_{d}^{\infty}(\F_{t},\C_{t}\cap\D_t).
\end{equation}
Hence, an investor with a financial position $Y\in L_{d}^{\infty}$ can ideally reach any element of the set $Y+\K$ by trading in the market. However, it may be the case that the risk of the resulting position is evaluated only through a (small) selection of the $d$ assets, in other words, trading has to be done in such a way that the only possibly nonzero components of the resulting position can be in some selected subset of the $d$ assets. Without loss of generality, suppose that liquidation is made into the first $m\leq d$ of the assets. The idea of liquidation is made precise by the notion of liquidation function introduced in Definition~\ref{liquidfunc} below.
Let us introduce the linear operator $B\colon\R^{m}\to\R^{d}$ defined by
\begin{equation}\label{linop}
Bz=(z_{1},\ldots,z_{m},0,\ldots,0)^{\mathsf{T}}.
\end{equation}
We will use the composition of $B$ with random variables in $L_{m}^{0}$. Given $X\in L_{m}^{0}$, $BX$ denotes the element in $L_{d}^{0}$ defined by $(BX)(\omega)=B(X(\omega))$ for $\omega\in\Omega$. The adjoint $B^{*}:\R^{d}\rightarrow\R^{m}$ of $B$ is given by
\begin{equation}\label{adjop}
B^{*}y=(y_{1},\ldots,y_{m})^{\mathsf{T}}.
\end{equation}
Similarly, $B^{*}$ can be composed with random variables in $L_{d}^{0}$. With a slight abuse of notation, we will also use $B^{*}$ in the context of vector probability measures. Given $\mathbb{Q}\in\mathcal{M}_{d}(\mathbb{P})$, we define $B^{*}\mathbb{Q}=(\mathbb{Q}_{1},\ldots,\mathbb{Q}_{m})^{\mathsf{T}}\in\mathcal{M}_{m}(\mathbb{P})$.

\begin{defn}\label{liquidfunc}
The function $\Lambda_{m} \colon L^{\infty}_{d} \to \P(L^\infty_m)$ defined by
\begin{equation}
\Lambda_{m}(Y) = \cb{X \in L^{\infty}_{m} \mid BX \in Y + \K}
\end{equation}
is called the \emph{liquidation function} associated with $\K$.
\end{defn}

Hence, given $Y\in L_{d}^{\infty}$, the set $\Lambda_{m}(Y)$ consists of all possible resulting positions in $Y+\K$ that are already liquidated into the first $m$ assets.

\subsection{Market risk measures and their dual representations}\label{marketmain}

Let us consider a closed convex risk measure $R \colon L_{m}^{\infty}\to\G_m$ which is used for risk evaluation after liquidating the resulting positions into the first $m$ assets. As all the positions in $\Lambda_m(Y)$ are accessible to the investor with position $Y\in L_d^\infty$, the value of $R$ is to be minimized over the set $\Lambda_{m}(Y)$ as the following definition suggests.

\begin{defn}\label{marketextension}
The function $R^{\market}: L_{d}^{\infty}\to\P(\R^{m})$ defined by
\begin{equation}
R^{\market}(Y)\coloneqq \inf_{(\G_m,\supseteq)}\cb{R(X) \mid X \in \Lambda_{m}(Y)} =\cl\co\bigcup_{X\in\Lambda_{m}(Y)}R(X),
\end{equation}
is called the \emph{market risk measure} induced by $R$.
\end{defn}

\begin{rem}
In the case of the conical market model described in Example~\ref{conical}, when $\D_{t}(\o)=\R^{d}$ for each $\o\in\O$ and $t\in\{0,\ldots,T\}$, and 
no liquidation at $t=T$ is considered \mbox{($m=d$),} Definition~\ref{marketextension} recovers the notion of market-extension (with closed values) given in \citet[Definition~2.8, Remark~2.9]{hry:avar}.
\end{rem}

Recall that a closed convex risk measure $R\colon L_m^\infty\to\G_m$ is defined by the five properties in Proposition~\ref{set-srm-lsc}. For the market risk measure, these properties need to be rewritten with obvious changes as the function is now defined on $L_d^\infty$. (For instance, the translativity of $R^{\market}$ reads as $R^{\market}(Y+Bz)=R^{\market}(Y)-z$ for every $Y\in L_d^\infty$ and $z\in\R^m$.) The next proposition shows that the market risk measure is a closed convex risk measure except for a finiteness condition and weak$^\ast$-closedness.

\begin{prop}\label{markettransition}
The market risk measure $R^{\market}$ is monotone, translative and convex, and it satisfies $R^{\market}(0)\neq \emptyset$. In addition, the convex hull operator can be dropped from Definition~\ref{marketextension}, that is, for $Y\in L_d^{\infty}$,
\begin{equation}\label{co-dropped}
R^{\market}(Y)=\cl\bigcup_{X\in\Lambda_{m}(Y)}R(X).
\end{equation}
\end{prop}

To recover weak$^\ast$-closedness, we define the closed version of $R^{\market}$ via the notion of \emph{closed hull}.

\begin{defn}\label{closedhull}
The \emph{closed hull} $\cl F$ of a function $F \colon L_{d}^{\infty}\to\G_m$ is the pointwise greatest weak$^\ast$-closed function minorizing it, that is, if $F^{\prime} \colon L_{d}^{\infty}\to\G_m$ is a weak$^{*}$-closed function such that $F(Y)\subseteq F^{\prime}(Y)$ for all $Y\in L_{d}^{\infty}$, then we have $(\cl F)(Y)\subseteq F^{\prime}(Y)$ for every $Y\in L_{d}^{\infty}$. The closed hull $\cl R^{\market}$ of $R^{\market}$ is called the \emph{closed market risk measure} induced by $R$.
\end{defn}

One can check that monotonicity, translativity and convexity are preserved under taking the closed hull. Hence, in view of Proposition~\ref{markettransition}, the closed market risk measure induced by a closed convex risk measure $R \colon L_{m}^{\infty}\rightarrow\G_m$ is a closed convex risk measure if $(\cl R^{\market})(0)\neq\R^m$. Theorem~\ref{liquidationtheorem} below gives a dual representation of the closed market risk measure in terms of the penalty function of the original risk measure $R$ under the assumption of finiteness at zero. The special case of no trading constraints in a convex (conical) market model is given in Corollary~\ref{corollary_convex market} (Corollary~\ref{liquidationcorollary}). The set of dual variables to be used in the results below is given by
\begin{equation}
\mathcal{W}_{m,d}\coloneqq\mathcal{M}_{d}(\mathbb{P})\times ((\R^{m}_{+}\bs\{0\})\times\R^{d-m}_{+}).
\end{equation}
We will also make use of the homogeneous halfspaces $G(w)\coloneqq\{y\in\R^{d}\mid w^{\mathsf{T}}y\geq 0\}$ for $w\in\R^{d}_{+}\bs\{0\}$.

\begin{thm}\label{liquidationtheorem}
Suppose that $R\colon L_{m}^{\infty}\to\G_m$ is a closed convex risk measure with penalty function $-\alpha_{R}\colon\M_{m}(\Pr)\times(\R^{m}_{+}\bs\{0\})\to\G_m$, see Proposition~\ref{dualthm}. Assume that $(\cl R^{\market})(0)\neq\R^m$. Then the closed market risk measure $\cl R^{\market}\colon L_{d}^{\infty}\to\G_m$ is also a closed convex risk measure, and it has the following dual representation: For every $Y\in L_{d}^{\infty}$,
\begin{align}
(\cl R^{\market})(Y) =\bigcap_{(\Q,w)\in \mathcal{W}_{m,d}}\sqb{-\alpha_{\cl R^{\market}}(\Q,w)+B^{*}\of{\of{\E^{\Q}\sqb{-Y}+G(w)}\cap B(\R^{m})}},
\end{align}
where $-\alpha_{\cl R^{\market}} \colon \mathcal{W}_{m,d}\to\G_m$ is defined by
\begin{align}\label{alph}
\alpha_{\cl R^{\market}}(\Q,w)=&-\alpha_{R}(B^{*}\Q,B^{*}w)\notag \\
& +\sum_{t=0}^{T}\cl\bigcup_{U^{t}\in L_{d}^{\infty}(\F_{t}, \C_{t}\cap \D_{t})} B^{*}\of{\of{\E^{\Q}\sqb{U^{t}}+G(w)}\cap B(\R^{m})}.
\end{align}
\end{thm}

Recall that the \emph{recession cone} of a nonempty convex set $C\subseteq\R^{d}$ is the convex cone $0^{+}C\coloneqq \{y\in\R^{d}\mid y+C\subseteq C\}$ and the \emph{positive dual cone} of a nonempty convex cone $K\subseteq\R^{d}$ is the convex cone $K^{+}\coloneqq\{y\in\R^{d}\mid\forall k\in K\colon y^{\mathsf{T}}k\geq 0\}$; see \citet[Section~8, p.~61]{rockafellar} and \citet[Section~1.1, p.~7]{zalinescu}, for instance.

\begin{cor}\label{corollary_convex market}
Under the assumptions of Theorem~\ref{liquidationtheorem}, suppose that $\mathcal{D}_{t}(\o)=\R^{d}$ for each $\o\in\O$ and $t\in\T$. Then $-\alpha_{\cl R^{\market}}$ given by \eqref{alph} is concentrated on the set
\begin{equation}
\mathcal{W}_{m,d}^{\convex}\coloneqq\Big\{(\mathbb{Q},w)\in\mathcal{W}_{m,d}\mid \forall t\in\T\colon w\cdot \mathbb{E}\sqb{\frac{d\mathbb{Q}}{d\mathbb{P}}\;\middle\vert\;\mathcal{F}_{t}}\in L_{d}^{1}(\mathcal{F}_{t},(0^{+}\mathcal{C}_{t})^{+}) \Big\},
\end{equation}
where, for each $t\in\T$, $(0^{+}\mathcal{C}_{t})^{+}\colon\Omega\rightarrow\G_d$ is the measurable function defined by $(0^{+}\mathcal{C}_{t})^{+}(\omega)\coloneqq (0^{+}\mathcal{C}_{t}(\omega))^{+}$. 
\end{cor}

In other words, we have $-\alpha_{\cl R^{\market}}(\mathbb{Q},w)=\R^{m}$ for $(\mathbb{Q},w)\in\mathcal{W}_{m,d}\bs\mathcal{W}_{m,d}^{\convex}$ within the setting of the previous result.

\begin{cor}\label{liquidationcorollary}
Under the assumptions of Theorem~\ref{liquidationtheorem} suppose that $\mathcal{D}_{t}(\o)=\R^{d}$ for each $\o\in\O$ and $t\in\T$, and that the market model is conical as in Example~\ref{conical}. Consider the set
\begin{equation}
\mathcal{W}_{m,d}^{\cone}\coloneqq\Big\{(\mathbb{Q},w)\in\mathcal{W}_{m,d}\mid \forall t\in\T\colon w\cdot\mathbb{E}\sqb{\frac{d\mathbb{Q}}{d\mathbb{P}}\;\middle\vert\;\mathcal{F}_{t}}\in L_{d}^{1}(\mathcal{F}_{t},\mathcal{C}_{t}^{+}) \Big\},
\end{equation}
where, for each $t\in\T$, $\mathcal{C}_{t}^{+}\colon\Omega\rightarrow\G_d$ is the measurable function defined by
$\mathcal{C}_{t}^{+}(\omega)\coloneqq(\mathcal{C}_{t}(\omega))^{+}$. Then, \eqref{alph} reduces to
\begin{equation}\label{penalty}
-\alpha_{\cl R^{\market}}(\mathbb{Q},w)=\begin{cases}-\alpha_{R}(B^{*}\mathbb{Q},B^{*}w) & \text{if }(\mathbb{Q},w)\in\mathcal{W}_{m,d}^{\cone},\\ \R^{m} & \text{else}, \end{cases}
\end{equation}
for each $(\mathbb{Q},w)\in\mathcal{W}_{m,d}$; hence, for every $Y\in L_{d}^{\infty}$,
\begin{align}
(\cl R^{\market})(Y)=\bigcap_{(\mathbb{Q},w)\in \mathcal{W}_{m,d}^{\cone}}\sqb{-\alpha_{R}(B^{*}\mathbb{Q},B^{*}w)+B^{*}\left(\left(\mathbb{E}^{\mathbb{Q}}\sqb{-Y}+G(w)\right)\cap B(\R^{m})\right)}.
\end{align}
\end{cor}

The proofs of Theorem~\ref{liquidationtheorem}, Corollary~\ref{corollary_convex market}, Corollary~\ref{liquidationcorollary} above are given in Section~\ref{marketproof}. They rely on the observation that, roughly speaking, the market risk measure is the \emph{(set-valued) infimal convolution} of the original risk measure and the (set-valued) \emph{indicator functions} of the convex sets $L_{d}^{\infty}(\mathcal{F}_{t},\mathcal{C}_{t}\cap \mathcal{D}_{t})$, $t\in\T$. This technical observation is discussed in Section~\ref{marketproof}, where the definitions of these notions are also given.

\subsection{Market risk measures induced by shortfall and divergence risk measures}\label{marketconditions}

In this section, we present sufficient conditions that guarantee the finite-valuedness condition $(\cl R^{\market})(0)\neq \R^m$ for the closed market risk measures induced by shortfall and divergence risk measures. Once this property is established, these closed market risk measures are closed convex risk measures and their dual representations are provided by Theorem~\ref{liquidationtheorem}. For simplicity, we assume that the market model is conical in the sense of Example~\ref{conical}.

\begin{assumption}\label{boundedmarket}
Suppose that the solvency cones of the market model share a common supporting halfspace in the sense that there exists $\bar{w}\in\R^{d}_{+}\bs\{0\}$ such that for $\mathbb{P}$-almost every $\omega\in\Omega$ and every $t\in\T$, $\inf_{y\in \mathcal{C}_{t}(\omega)}\bar{w}^{\mathsf{T}}y>-\infty$, or equivalently,
$\bar{w}\in (\mathcal{C}_{t}(\omega))^{+}$.
\end{assumption}

\begin{rem}
Assumption~\ref{boundedmarket} states the existence of a halfspace $G(\bar{w})=\{z\in\R^{d}\mid \bar{w}^{\mathsf{T}}z\geq 0\}$ for some $\bar{w}\in\R^{d}_{+}$ which satisfies $G(\bar{w})\supseteq \C_{t}(\omega)$ for $\Pr$-almost every $\o\in\O$ and $t\in\T$. In particular, when the solvency cones are constructed from bid-ask prices (see \citealt{kabanov}), this is equivalent to the ask prices having a uniform (in time and outcome) lower bound, or equivalently, the bid prices having a uniform (in time and outcome) upper bound. That is, $\bar{w}_j\leq\pi_{ij}(\o,t)\bar{w}_i$ for every $i,j\in\{1,\ldots,d\}$, every $t\in\T$, and $\Pr$-almost every $\o\in\O$, where $\pi_{ij}(\o,t)$ is the number of units of asset $i$ for which an agent can buy one unit of asset $j$ at time $t$ and state $\o$, and thus, denotes the ask price of asset $j$ in terms of asset $i$.
\end{rem}

\begin{prop}\label{finitedivergencerm}
Suppose that Assumption~\ref{boundedmarket} holds and $\dom\ell=\R^m$. 
\begin{enumerate}[(i)]
\item Let $r\in \R^m_{++}$ with $\frac{1}{r}\in\dom g$. If
\begin{equation}\label{commonsupport}
\inf_{x\in C}\bar{w}^{\mathsf{T}}(r\cdot x)>-\infty,
\end{equation}
then $(\cl D_{\ell,r}^{\market})(0)\neq\R^m$. In particular, $\cl D_{\ell,r}^{\market}$ is a closed convex risk measure with a dual representation provided by Theorem~\ref{liquidationtheorem}.
\item If there exists $r\in \R^m_{++}$ with $\frac{1}{r}\in\dom g$ such that \eqref{commonsupport} holds, then $(\cl R_{\ell}^{\market})(0)\neq\R^m$. In particular, $\cl R_{\ell}^{\market}$ is a closed convex risk measure with a dual representation provided by Theorem~\ref{liquidationtheorem}.
\end{enumerate}
\end{prop}

\section{Proofs and technical remarks}\label{appendix}

\subsection{Proofs of the results in Section~\ref{scalartheory}}\label{scalarproof}

\begin{proof}[Proof of Proposition~\ref{srm-lsc}]
Monotonicity, translativity and convexity are trivial. Let $X\in L^{\infty}$. It holds $\ell(-\esssup{X}-s) \leq \E\sqb{\ell(-X-s)} \leq \ell(-\essinf{X}-s)$ for every $s\in\R$. Note that $\ell$ is strictly increasing on $\ell^{-1}(\interior\ell(\R))\coloneqq\{x\in\R\mid\ell(x)\in\interior\ell(\R)\}=(a,b),$ where $a\coloneqq\inf\{x\in\R\mid\ell(x)> \inf_{y\in\R}{\ell(y)}\}\in\R\cup\{-\infty\}$ and $b\coloneqq\sup\{x\in\R\mid\ell(x)<+\infty\} \in\R\cup\{+\infty\}$. Hence, the inverse $\ell^{-1}$ is well-defined as a function from $\interior\ell(\R)$ to $(a,b)$. It holds $\mathbb{E}\sqb{\ell(-X-s)}\leq 0$ for each $s\geq -\essinf{X}-\ell^{-1}(0)$, and $\E\sqb{\ell(-X-s)}>0$ for each $s<-\esssup{X}-\ell^{-1}(0)$. So $\rho_{\ell}(X)\in\R$. Besides, $\E\sqb{\ell(-X-\rho_{\ell}(X))}\leq 0$ since the restriction of $\ell$ on $\dom\ell$ is a continuous function. To show (weak$^{*}$-)lower semicontinuity, let $(X^{n})_{n\in\mathbb{N}}$ be a bounded sequence in $L^{\infty}$ converging to some $X\in L^{\infty}$ $\Pr$-almost surely. Then, using Fatou's lemma together with the fact that the restriction of $\ell$ on $\dom\ell$ is nondecreasing and continuous, we have
\begin{align}
\E\sqb{\ell\of{-X-\liminf_{n\rightarrow\infty}{\rho_{\ell}(X^{n})}}}&=\E\sqb{\ell\of{\liminf_{n\rightarrow\infty}{(-X^n-\rho_{\ell}(X^n))}}}\\
&\leq\liminf_{n\rightarrow\infty}{\E\sqb{\ell(-X^{n}-\rho_{\ell}(X^{n}))}}\leq 0.\notag
\end{align}
This implies the so-called \emph{Fatou property} of $\rho_{\ell}$, namely, that $\rho_{\ell}(X)\leq\liminf_{n\rightarrow\infty}{\rho_{\ell}(X^{n})}$. By \citet[Theorem 4.33]{fs:sf}, this is equivalent to the lower semicontinuity of $\rho_{\ell}$.
\end{proof}

\begin{proof}[Proof of Proposition~\ref{srm-simple-dual}]
Note that $s\mapsto\E\sqb{\ell(-X-s)}$ is a proper convex function on $\R$. Hence, by Definition~\ref{srm-def}, $\rho_{\ell}(X)$ is the optimal value of a convex minimization problem. The corresponding Lagrangian dual objective function $h$ on $\R_+$ is given by
\begin{equation}
h(\lambda)= \inf_{s\in\R\colon\E\sqb{\ell(-X-s)}<+\infty}\of{s+\lambda\E\sqb{\ell(-X-s)}}.
\end{equation}
Clearly, $h(\lambda)=\delta_{\ell,\lambda}(X)$ if $\lambda>0$ since $\lambda\E\sqb{\ell(-X-s)}=+\infty$ if $\E\sqb{\ell(-X-s)}=+\infty$. On the other hand, note that $\E\sqb{\ell(-X-s)}<+\infty$ if and only if $\Pr\cb{-X-s\in \dom\ell}=1$. It follows that
\begin{equation}
h(0)=\inf\cb{s\in\R\mid \E\sqb{\ell(-X-s)}<+\infty}=-\essinf X - \sup \dom \ell.
\end{equation}
Therefore, the optimal value of the dual problem equals the right hand side of \eqref{srm-dual}. Finally, the two sides of \eqref{srm-dual} are equal since the usual \emph{Slater's condition} holds: There exists $\bar{s}\in\R$ such that $\E\sqb{\ell(-X-\bar{s})}<0$. This is because we have $\E\sqb{\ell(-X-s)}< 0$ for each $s> -\essinf{X}-\ell^{-1}(0)$, where $\ell^{-1}$ is the inverse function on $\interior\ell(\R)$ as in the proof of Proposition~\ref{srm-lsc}.

\end{proof}

\begin{proof}[Proof of Proposition~\ref{conjugation-bijection}]
Let $f$ be a loss function and $f^{*}:\R\to\R\cup\{+\infty\}$ its conjugate function. Note that $\dom f^{*}\subseteq\R_{+}$ since, for each $y<0$, we have
\begin{equation}
f^{*}(y)\geq\sup_{n\in\mathbb{N}}(-ny-f(-n))\geq\sup_{n\in\mathbb{N}}(-ny)-f(0)=+\infty,
\end{equation}
where we use the monotonicity of $f$ for the second inequality. Moreover, $0\in\dom f^{*}$ since $f^{*}(0)=-\inf_{x\in\R}f(x)<+\infty$. Clearly, $f^{*}(y)\geq -f(0)$ for each $y\in\R$. Besides, by \citet[Theorem~23.3]{rockafellar}, the subdifferential $\partial f(0)$ of $f$ at $0$ is nonempty and, by \citet[Theorem~23.5]{rockafellar}, we have $f^{*}(y)=-f(0)$ for every $y\in\partial f(0)$. Hence, $f^{*}$ attains its infimum. Finally, $f^{*}$ is not of the form $y\mapsto +\infty\cdot 1_{\{y<0\}}+(ay+b)\cdot 1_{\{y\geq 0\}}$ for some $a\in\R_{+}\cup\{+\infty\}$ and $b\in\R$ as otherwise we would get $f(x)=(f^{*})^{*}(x)=+\infty\cdot 1_{\{x>a\}}-b\cdot 1_{\{x\leq a\}}$, $x\in\R$, so that $f$ would be identically constant on $\dom f$. Hence, $f^{*}$ is a divergence function. Conversely, let $\varphi$ be a divergence function and $\varphi^{*}:\R\rightarrow\R\cup\{+\infty\}$ its conjugate function. Let $x^{1},x^{2}\in\R$ with $x^{1}\geq x^{2}$. Since $\dom \varphi\subseteq\R_{+}$, we have $x^{1}y-\varphi(y)\geq x^{2}y-\varphi(y)$ for each $y\in\dom \varphi$ so that $\varphi^{*}(x^{1})\geq \varphi^{*}(x^{2})$. Hence, $\varphi^{*}$ is nondecreasing. Moreover, $\inf_{x\in\R} \varphi^{*}(x)=-\varphi(0)>-\infty$ since $\varphi=(\varphi^{*})^{*}$ and $0\in\dom\varphi$. Clearly, $\varphi^{*}(0)=-\inf_{y\in\R}\varphi(y)\in\R$ so that $0\in\dom \varphi^{*}$. Finally, $\varphi^{*}$ is not identically constant on $\dom \varphi^{*}$ as otherwise $\varphi=(\varphi^{*})^{*}$ would fail to satisfy property (iii) in Definition~\ref{divergence}. Hence, $\varphi^{*}$ is a loss function.
\end{proof}

\begin{proof}[Proof of Theorem~\ref{oce-duality}]
If $\lambda=0$, then $I_{g,\lambda}(\Q\mid \Pr)=\sup\dom\ell$ for every $\Q\in\M(\Pr)$ and we have
\begin{equation}
\delta_{\ell,0}(X)=-\essinf X - \sup\dom\ell=\sup_{\Q\in\M(\Pr)}\E^\Q\sqb{-X}-\sup\dom\ell
\end{equation}
by the dual representation of the worst-case risk measure $X\mapsto - \essinf X$; see \citet[Example~4.39]{fs:sf}, for instance. Hence, $\eqref{oce-primal-dual}$ holds in the case $\lambda=0$. Moreover, by Remark~\ref{divergenceindex}, $\sup\dom\ell<+\infty$ if and only if $1\in \dom g_0$. Hence, $\delta_{\ell,0}$ is a lower semicontinuous convex risk measure if $1\in\dom g_0$, and $\delta_{\ell,0}(X)=-\infty$ for every $X\in L^\infty$ otherwise.

Suppose $\lambda>0$. Note that the right hand side of \eqref{oce-primal-dual} can be rewritten as a maximization problem on the space $L^{1}_+$ of integrable nonnegative real-valued random variables on $(\O,\mathcal{F},\Pr)$ (identified up to almost sure equality):
\begin{align}
&\sup_{\Q\in\M(\Pr)} \of{\E^{\Q}\sqb{-X}-I_{g,\lambda}(\Q\mid\Pr)} \notag \\
&=
\sup_{V\in L^{1}_{+}} \cb{\E\sqb{-XV}-\lambda\E\sqb{g\of{\frac{1}{\lambda}V}} \mid\E\sqb{V}=1}.
\end{align}
The optimal value of the corresponding Lagrangian dual problem is computed as
\begin{align}
q_{X}&\coloneqq\inf_{s\in\R}\sup_{V\in L^{1}_{+}}\of{\E\sqb{-XV}-\lambda\E\sqb{g\of{\frac{1}{\lambda}V}} + s(1-\E\sqb{V})} \\
&=\inf_{s\in\R}\of{s+\sup_{V\in L^{1}_{+}}\E\sqb{(-X-s)V-\lambda g\of{\frac{1}{\lambda}V}}}\notag \\
&=\inf_{s\in\R}\of{s+\E\sqb{\sup_{z\in\R_{+}}\of{(-X-s)z - \lambda g\of{\frac{1}{\lambda}z}}}} \notag \\
&=\inf_{s\in\R}(s+\E\sqb{g_{\lambda}^{*}(-X-s)}),\notag
\end{align}
where the third equality is due to \citet[Theorem~14.60]{rockafellar2}, and $g_{\lambda}^{*}$ is the conjugate of the divergence function $g_{\lambda}$; see Remark~\ref{divergenceindex}. Hence, $g_{\lambda}^{*}=\lambda\ell$ and $q_{X}$ equals the left hand side of \eqref{oce-primal-dual}. Finally, to conclude \eqref{oce-primal-dual}, we consider the following cases:

\begin{enumerate}[(i)]
\item Suppose that $1\in\interior\dom g_{\lambda}$, that is, $\frac{1}{\lambda}<\beta$. (Recall that $\interior\dom g_\lambda=(0,\lambda\beta)$, see Definition~\ref{divergencedefn} \emph{et seq.}) Then the following constraint qualification holds, for instance, with $\bar{V}\equiv 1$:
\begin{equation}\label{CQ}
\exists \bar{V}\in L^{1}_{+}:\text{  }\E\sqb{\bar{V}}=1,\text{ }\bar{V}\in\interior\dom g_{\lambda}\;\;\Pr\text{-almost surely}.
\end{equation}
By \citet[Corollary~4.8]{bl:pfcp}, \eqref{CQ} suffices to conclude \eqref{oce-primal-dual}. Note that we have
\begin{align}
-\E\sqb{X}-\lambda g\of{\frac{1}{\lambda}}&\leq\sup_{V\in L^{1}_{+}}\cb{\E\sqb{-XV}-\lambda\E\sqb{g\of{\frac{1}{\lambda}V}}\mid \E\sqb{V}=1} \\
&\leq -\essinf{X}-\lambda\inf_{x\in\R}{g(x)}\notag
\end{align}
so that both sides of \eqref{oce-primal-dual} are in $\R$.

\item Suppose that $\lambda\beta=1$ and $\dom g_\lambda=[0,\lambda\beta]=[0,1]$, that is, $\dom g=[0,\beta]=[0,\frac{1}{\lambda}]$. In this case, the only $V\in L^{1}_{+}$ with $\E\sqb{V}=1$ and $\Pr\cb{V\in\dom g_{\lambda}}=1$ is $V\equiv 1$, and hence, the right hand side of \eqref{oce-primal-dual} gives $-\E\sqb{X}-\lambda g(\frac{1}{\lambda})\in\R$. Note that \eqref{CQ} fails to hold here. Using \eqref{oce-primal-dual} for the previous case, we have
\begin{align}
\inf_{s\in\R}(s+\lambda\E\sqb{\ell(-X-s)}) &=\lim_{\varepsilon\downarrow 0}\inf_{s\in\R}(s+(\lambda+\varepsilon)\E\sqb{\ell(-X-s)})\\
&=\lim_{\varepsilon\downarrow 0} \sup_{\Q\in\M(\Pr)}\of{\E^{\Q}\sqb{-X}-(\lambda+\varepsilon)\E\sqb{g\of{\frac{1}{\lambda+\varepsilon}\frac{d\Q}{d\Pr}}}},\notag
\end{align}
where the first equality follows since the proper, concave, upper semicontinuous function
\begin{equation}
\R\ni \gamma\mapsto\inf_{s\in\R}(s+ \gamma\E{\ell(-X-s)})\in\R\cup\{-\infty\}\label{indexfunction}
\end{equation}
is right-continuous at $\gamma=\lambda$. Finally, we have
\begin{align}
&\lim_{\varepsilon\downarrow 0} \sup_{\Q\in\M(\Pr)}\of{\E^{\Q}\sqb{-X}-(\lambda+\varepsilon)\E\sqb{g\of{\frac{1}{\lambda+\varepsilon}\frac{d\Q}{d\Pr}}}} \\
&=\lim_{\varepsilon\downarrow 0} \sup_{\Q\in\M(\Pr)}\of{\E^{\Q}\sqb{-X}-(\lambda+\varepsilon)\E\sqb{g\of{\frac{1}{\lambda+\varepsilon}\frac{d\Q}{d\Pr}}-g(0)}} \notag \\
&\quad -\lim_{\varepsilon\downarrow 0}(\lambda+\varepsilon)g(0)\label{twosides}\\
&=\inf_{\gamma\in[\lambda,\lambda+\varepsilon^\prime]} \sup_{\Q\in\M(\Pr)}\of{\E^{\Q}\sqb{-X}-\gamma\E\sqb{g\of{\frac{1}{\gamma}\frac{d\Q}{d\Pr}}-g(0)}} - \lambda g(0)\notag\\
&=\sup_{\Q\in\M(\Pr)}\inf_{\gamma\in[\lambda,\lambda+\varepsilon^\prime]}\of{\E^{\Q}\sqb{-X}-\gamma\E\sqb{g\of{\frac{1}{\gamma}\frac{d\Q}{d\Pr}}-g(0)}} - \lambda g(0)\notag \\
&=\sup_{\Q\in\M(\Pr)}\of{\E^{\Q}\sqb{-X}-\lambda\E\sqb{g\of{\frac{1}{\lambda}\frac{d\Q}{d\Pr}}}}\notag \\
&=-\E\sqb{X}-\lambda g\of{\frac{1}{\lambda}}\in\R,\notag
\end{align}
where $\varepsilon^{\prime}>0$ is some fixed number. Here, the second equality follows since the function $\gamma\mapsto \gamma (g(\frac{y}{\gamma})-g(0))$ is a nonincreasing function on $\R_{++}$ for each $y\in\R$; see Remark~\ref{divergenceindex}. The third equality is due to a classical minimax theorem and it uses the compactness of the interval $[\lambda,\lambda+\varepsilon^\prime]$, see \citet[Corollary~3.3]{sion}. The fourth equality follows by monotone convergence theorem and the monotonicity of the function $\gamma\mapsto \gamma (g(\frac{y}{\gamma})-g(0))$ on $\R_{++}$. The last equality is already discussed above. Finally, the first equality follows since the two limits in \eqref{twosides} are shown to be finite by the succeeding equalities. Hence, we obtain \eqref{oce-primal-dual}.

\item Suppose $1\notin\dom g_\lambda$, that is, either $\dom g_\lambda=[0,\lambda\beta)$ and $\lambda\beta\geq 1$, or, $\dom g_\lambda=[0,\lambda\beta]$ and $\lambda\beta>1$. In this case, there is no $Y\in L^{1}_{+}$ with $\E\sqb{Y}=1$ and $\Pr\cb{Y\in\dom g_{\lambda}}=1$. Hence, the right hand side of \eqref{oce-primal-dual} gives $-\infty$. On the other hand, we have
\begin{align}\label{boundednessinequality}
\inf_{s\in\R}(s+\lambda\E\sqb{\ell(-X-s)}) &\leq\inf_{s\in\R}(s+\lambda\ell(-\essinf{X}-s))\\ &=-\essinf{X}-\sup_{s\in\R}(s-\lambda\ell(s))\notag \\
& = -\essinf{X}-\lambda g\of{\frac{1}{\lambda}}=-\infty.\notag 
\end{align}
\end{enumerate}
Hence, \eqref{oce-primal-dual} holds. In the first two cases where $1\in\dom g_\lambda$, we observe that $\delta_{\ell,\lambda}(0)\in\R$. Moreover, \eqref{divergencerisk} guarantees monotonicity, translativity, convexity and lower semicontinuity directly, which makes $\delta_{\ell,\lambda}$ a lower semicontinuous convex risk measure. In the last case where $1\notin\dom g_\lambda$, $\delta_{\ell,\lambda}(X)= -\infty$ for every $X\in L^\infty$.
\end{proof}

\begin{proof}[Proof of Proposition~\ref{divergence-minimal}]
Let $\Q\in\M(\Pr)$ and $\lambda\in\R_+$ with $1\in \dom g_\lambda$. If $\lambda=0$, then it follows from Definition~\ref{divergence} and the proof of Theorem~\ref{oce-duality} that $\alpha_{\delta_{\ell,0}}(\Q)=\sup\dom\ell=I_{g,0}(\Q\mid\Pr)$. Suppose $\lambda>0$. Using \eqref{divergencerisk} and the definition of penalty function in \eqref{dual-srm},
\begin{align}
\alpha_{\delta_{\ell,\lambda}}(\Q)&=\sup_{s\in\R}\of{-s+\sup_{X\in L^{\infty}}\E\sqb{-\frac{d\Q}{d\Pr}X-\lambda\ell(-X-s)}}\\
&=\sup_{s\in\R}\of{-s+\E\sqb{\sup_{x\in \R}\of{-\frac{d\Q}{d\Pr}x-\lambda\ell(-x-s)}}}\notag\\
&=\sup_{s\in\R}\of{-s+\E\sqb{\frac{d\Q}{d\Pr}s+g_{\lambda}\of{\frac{d\Q}{d\Pr}}}}\notag\\
&=I_{g,\lambda}(\Q\mid\Pr),\notag 
\end{align}
where the second equality follows from \citet[Theorem~14.60]{rockafellar2} and the third equality follows from Remark~\ref{divergenceindex}. For the penalty function of $\rho_{\ell}$, note that
\begin{align}
\alpha_{\rho_{\ell}}(\Q)&=\sup_{X\in L^{\infty}}\of{\E^{\Q}\sqb{-X}-\inf_{s\in\R}\of{s+I_{(-\infty,0]} \of{\E\sqb{\ell(-X-s)}}}}\\
&=
\sup_{X\in L^{\infty}}\of{\E^{\Q}\sqb{X}-I_{(-\infty,0]}(\E\sqb{\ell(X)})}\notag \\
&=\sup_{X\in L^{\infty}}\cb{\E^{\Q}\sqb{X}\mid \E\sqb{\ell(X)}\leq 0}.\notag 
\end{align}
For the last maximization problem, the corresponding Lagrangian dual objective function $h$ on $\R_+$ is given by
\begin{equation}
h(\lambda)=\sup_{X\in L^{\infty}\colon\E\sqb{\ell(X)}<+\infty}\of{\E^{\Q}\sqb{X}-\lambda\E\sqb{\ell(X)}}.
\end{equation}
Note that $\E\sqb{\ell(X)}<+\infty$ if and only if $\Pr\cb{X\in\dom\ell}=1$. If $\lambda=0$, then
\begin{equation}
h(0)=\sup_{X\in L^{\infty}\colon\E\sqb{\ell(X)}<+\infty}\E^{\Q}\sqb{X}=\sup\dom\ell=I_{g,0}\of{\Q\mid\Pr}.
\end{equation}
On the other hand, if $\lambda>0$, then
\begin{align}
h(\lambda)&=\sup_{X\in L^{\infty}\colon\E\sqb{\ell(X)}<+\infty}\of{\E^{\Q}\sqb{X}-\lambda\E\sqb{\ell(X)}}\\
&=\E\sqb{\sup_{x\in \R}\of{\frac{d\Q}{d\Pr}x-\lambda\ell(x)}}=I_{g,\lambda}\of{\Q\mid\Pr},\notag 
\end{align}
where we use \citet[Theorem~14.60]{rockafellar2} for the second equality and Remark~\ref{divergenceindex} for the third equality. Hence, the optimal value of the dual problem is given by
\begin{equation}
q(\Q)\coloneqq\inf_{\lambda\in\R_+}h(\lambda)=\inf_{\lambda\in\R_+}I_{g,\lambda}\of{\Q\mid\Pr}.
\end{equation}
 Note that Slater's condition holds, that is, there exists $\bar{X}\in L^{\infty}$ such that $\E\sqb{\ell(\bar{X})}<0$; take, for example, $\bar{X}\equiv\ell^{-1}(0)-1$, where $\ell^{-1}$ is the inverse function on $\interior\ell(\R)$ as in the proof of Proposition~\ref{srm-lsc}. Therefore, $\alpha_{\rho_{\ell}}(\Q)=q(\Q)$. Note that $I_{g,\lambda}(\Q\mid\Pr)=+\infty$ for every $\lambda\in\R_+$ with $1\notin \dom g_{\lambda}$, see case (iii) in the proof of Theorem~\ref{oce-duality}. Hence, we also have $q(\Q)=\inf_{\lambda\in\R_+\colon1\in \dom g_\lambda}\alpha_{\delta_{\ell,\lambda}}(\Q)$.
\end{proof}

\subsection{A remark about the scalar loss functions}\label{technical}

In \citet[Theorem~10]{fs:srm} and \citet[Theorem~4.115]{fs:sf}, the second part of Proposition~\ref{divergence-minimal} is proved with the additional assumption that $\ell$ maps into $\R$. This assumption implies that the $\ell$-shortfall risk measure is continuous from below and the first supremum in \eqref{dual-srm} is attained (\citealt[Proposition~4.113]{fs:sf}). Besides, the same assumption implies the so-called \emph{superlinear growth condition} on $g$, namely, that $\lim_{y\rightarrow\infty}\frac{g(y)}{y}=+\infty$ (\citealt[Lemma~11]{fs:srm}). The analytic proof for Proposition~\ref{divergence-minimal} in \cite{fs:srm} makes use of this property instead of the dual relationship with divergence risk measures. Using this proposition and assuming that $1\in\dom g$, Theorem~\ref{oce-duality} is proved for $\lambda=1$ (\citealt[Theorem~4.122]{fs:sf}), in which case $\delta_{\ell,1}$ is guaranteed to be a risk measure (it has finite values). In our treatment, while $1$ may not be in $\dom g$, there exists some $\bar{\lambda}>0$ with $1\in \dom g_{\bar{\lambda}}$ and hence $\delta_{\ell,\bar{\lambda}}$ is a risk measure.

In \cite{bt:oce}, on the other hand, the divergence function $g$ is of central importance: In addition to the assumptions here, it is assumed in \cite{bt:oce} that $g$ attains its infimum at $1$ with value $0$, which is equivalent to assuming that $\ell(0)=0$ and $1\in\partial\ell(0)$. These assumptions make $g$ a natural divergence function in the sense that the function $\Q\mapsto\E[g(\frac{d\Q}{d\Pr})]$ on $\M(\Pr)$ has nonnegative values and takes the value $0$ if $\Q=\Pr$; $\E[g(\frac{d\Q}{d\Pr})]$ can be interpreted as the distance between some ``subjective'' measure $\Q\in\M(\Pr)$ and the physical measure $\Pr$. On the other hand, the additional assumptions on the loss function $\ell$ may be considered as restrictive. Here, we take $\ell$ as the central object by dropping these assumptions and use the convex duality approach as in \cite{bt:oce}. Note that Theorem~\ref{oce-duality} (\citealt[Theorem~4.2]{bt:oce}) and the first part of Proposition~\ref{divergence-minimal} (\citealt[Theorem~4.4]{bt:oce}) are proved in \cite{bt:oce} for the case $\lambda=1$. Here, we generalize this proof, basically, by considering the cases where the constraint qualification \eqref{CQ}, which is also used in the proof of \citet[Theorem~4.2]{bt:oce}, fails to hold.

\subsection{Lagrange duality for set optimization: a quick review}\label{lagrange-app}

The proofs of Theorem~\ref{srmdrmconnection} and Proposition~\ref{set-srm-penaltyname} rely on the application of the recent Lagrange duality in \cite{hl:lagrange}. We quickly review the definition of the dual problem here. Let $\mathcal{X}$ be a locally convex topological linear space. Consider a set minimization problem of the form \eqref{setoptimization}, where $\Phi\colon\mathcal{X}\to\G_m$ is an arbitrary objective function and $\Psi\colon\mathcal{X}\to\G_m$ is an arbitrary constraint function. The optimal value of this problem is $p\coloneqq \inf_{(\G_m,\supseteq)}\cb{\Phi(x)\mid 0\in\Psi(x), x\in\mathcal{X}}$.

The halfspace-valued functions $S_{\lambda, v} \colon \R^m\to \G_m$ for $\lambda \in \R^m$, $v \in \R^m_+\bs\{0\}$ defined by
\begin{equation}
S_{\lambda,v}(z) = \cb{\eta \in\R^m \mid v^{\mathsf{T}}\eta \geq {\lambda}^{\mathsf{T}}z}
\end{equation}
will be used as set-valued substitutes for the (continuous) linear functionals of the scalar duality theory as in \cite{andreasduality}, \cite{hl:lagrange}. Here, there are two types of dual variables: The variable $\lambda \in \R^{m}$ is the usual vector of \emph{Lagrange multipliers} which is used to scalarize the values of $\Psi$ whereas the variable $v \in \R^{m}_{+}$ is the \emph{weight vector} which is used to scalarize the values of $\Phi$. The set-valued Lagrangian $L \colon \mathcal{X} \times \R^m \times (\R^m_+\bs\{0\}) \to \G_m$ and the objective function $H \colon \R^m \times (\R^m_+\bs\{0\}) \to \G_m$ of the dual problem for \eqref{setoptimization} are defined by
\begin{align}\label{lagrange}
&L(x, \lambda, v) =\cl\of{\Phi(x)+\inf_{(\G_m, \supseteq)}\cb{S_{\lambda, v}(z) \mid z \in \Psi(x)}},\\ 
&H(\lambda, v) =\inf_{(\G_m,\supseteq)} \{L(x, \lambda, v) \mid x \in \mathcal{X}\}.\notag 
\end{align}
The optimal value $q$ of the dual problem is the supremum of the dual objective function over the dual variables:
\begin{equation}\label{dualproblem}q\coloneqq\sup_{(\G_m, \supseteq)}\{H(\lambda, v) \mid \lambda\in\R^m, v \in \R^m_+\bs\{0\}\}.
\end{equation}
\begin{prop}\label{strong}
(\citealt[Theorem~6.1]{hl:lagrange}) Assume that $\Phi$ and $\Psi$ are convex functions and $p\neq\R^m$. Strong duality holds, that is, $p=q$ if the following set-valued version of Slater's condition holds: There exists $\bar{x}\in\mathcal{X}$ such that $\Phi(\bar{x})\neq\emptyset$ and $\Psi(\bar{x})\cap -\R^m_{++} \neq\emptyset$.
\end{prop}

\subsection{Proofs of the results in Section~\ref{mainsection}}\label{setvaluedproof}

\begin{proof}[Proof of Proposition~\ref{set-srm-lsc}]
Monotonicity, translativity and convexity are trivial. To show finiteness at $0$, using the proof of Proposition~\ref{srm-lsc}, we can find $z^{1}\in\R^{m}$ with $\ell(-z^{1})\in -\R_{+}^{m}$ and $z^{2}\in\R^{m}$ with $\ell(-z^{2})\in \R^m_{++}$. By the properties of $C$, it follows that $R_{\ell}(0)\not\in\{\emptyset,\R^{m}\}$. To show weak$^{\ast}$-closedness, let $(X^{n})_{n\in\mathbb{N}}$ be a bounded sequence in $L_{m}^{\infty}$ converging to some $X \in L_{m}^{\infty}$ $\Pr$-almost surely. Let $z \in \R^{m}$ and suppose that there exists $z^{n}\in R_{\ell}(X^{n})$, for each $n\in\mathbb{N}$, such that $(z^{n})_{n\in\mathbb{N}}$ converges to $z$. Using dominated convergence theorem, the closedness of $-C$, and the fact that the restriction of $\ell$ on $\dom\ell \coloneqq \{x \in \R^m \mid \ell(x) \in \R^{m}\} = \bigtimes_{i=1}^{m}\dom\ell_{i}\subseteq\R^{m}$ is continuous, we have
\begin{equation}
\E\sqb{\ell(-X - z)} =\E\sqb{\ell\of{\lim_{n\rightarrow\infty}{(-X^{n}-z^{n})}}} =\lim_{n\rightarrow\infty}{\E\sqb{\ell(-X^{n}-z^{n})}}\in -C,
\end{equation}
that is, $z \in R_{\ell}(X)$. This shows the so-called \emph{Fatou property} of $R_{\ell}$, namely, that
\begin{equation}
\liminf_{n\rightarrow\infty}R_{\ell}(X^{n}) \negthinspace\coloneqq\negthinspace\cb{z \in \R^m \negthinspace\mid \negthinspace\forall n \in \mathbb{N}\;\exists z^n \in R_{\ell}(X^n) \colon \negthinspace\lim_{n\rightarrow\infty}z^{n}=z} \negthinspace\subseteq\negthinspace R_{\ell}(X).
\end{equation}
By \citet[Theorem 6.2]{hh:duality}, this is equivalent to the weak$^{\ast}$-closedness of $R_{\ell}$.
\end{proof}

For the proof of Theorem~\ref{srmdrmconnection}, we will need the following lemmata.

\begin{lem}\label{set-srm-simple-dual}
For every $X\in L_{m}^{\infty}$,
\begin{equation}\label{set-srm-dual}
R_{\ell}(X)=\bigcap_{\substack{\lambda\in\R^m_+,\\v \in \R^m_+\bs\{0\}}}\cb{\eta \in \R^m \mid v^{\mathsf{T}}\eta \geq \inf_{\substack{z \in \R^m\colon\\\E\sqb{\ell(-X-z)}\in\R^m}}f_{\lambda,v}(z)+\inf_{x \in C}{{\lambda}^{\mathsf{T}}x}},
\end{equation}
where
\begin{equation}
f_{\lambda,v}(z)\coloneqq v^{\mathsf{T}}z+{\lambda}^{\mathsf{T}}\E\sqb{\ell(-X-z)}.
\end{equation}
\end{lem}

\begin{proof}
Let $X \in L^\infty_m$. Using \eqref{lagrange} and \eqref{infsup}, the Lagrangian for the problem \eqref{set-primal-shortfall} is computed as
\begin{align}
L(z,\lambda, v) & =\cl\of{z+\R_+^m+\cl\bigcup_{x \in \of{\E\sqb{\ell(-X-z)}+C} \cap \R^m}S_{\lambda, v}(x)}\\
& =\begin{cases} \cb{\eta \in \R^m \mid  v^{\mathsf{T}}\eta \geq f_{\lambda,v}(z)+\inf_{x \in C}{{\lambda}^{\mathsf{T}}x}}&\text{if }\E\sqb{\ell(-X-z)} \in \R^m,\\ \emptyset & \text{if }\E\sqb{\ell(-X-z)} \notin \R^m\end{cases}\notag 
\end{align}
for $z\in\R^m, \lambda\in\R^m, v\in\R^m_+\bs\cb{0}$. Hence, the dual objective function is given by
\begin{equation}\label{set-dual-halfspace}
H(\lambda, v) 
\negthinspace=\negthinspace \cb{\eta \in \R^m \negthinspace\mid\negthinspace  v^{\mathsf{T}}\eta \geq \inf_{z \in \R^m\colon \E\sqb{\ell(-X-z)} \in \R^m}f_{\lambda,v}(z) \negthinspace+\negthinspace\inf_{x \in C}{{\lambda}^{\mathsf{T}}x}}\negthinspace
\end{equation}
for $\lambda \in \R^m$, $v\in\R^m_+\bs\cb{0}$. Suppose $\lambda \not\in \R_+^m$. Since $C+\R_{+}^{m}\subseteq C$, there exists $\bar{x}\in C$ such that, for every $n \in \N$, we have $n\bar{x}\in C$ and $\lambda^{\mathsf{T}}\bar{x}<0$. Hence, $\inf_{x \in C}\lambda^{\mathsf{T}}x = -\infty$ and $H(\lambda, v) = \R^m$ for every $v \in \R^m_+\bs\{0\}$. Therefore, by \eqref{dualproblem} and \eqref{infsup}, the optimal value of the dual problem is given by the right hand side of \eqref{set-srm-dual}. Finally, the two sides of \eqref{set-srm-dual} are equal by Proposition~\ref{strong} since \emph{Slater's condition} holds: There exists $\bar{z} \in \R^m$ such that $(\E\sqb{\ell(-X-\bar{z})} + C)\cap -\R^m_{++} \neq \emptyset$. This follows as for the scalar version, see the proof of Proposition~\ref{srm-simple-dual}.
\end{proof}

\begin{lem}\label{changeofvariable}
Set $w^{\mathsf{T}}(-\infty) = -\infty$ whenever $w \in \R^{m}_{+}\bs\{0\}$. Then, for every $X\in L_{m}^{\infty}$,
\begin{equation}
R_{\ell}(X)=\bigcap_{r\in\R^m_+, w\in\R^{m}_{+}\bs\{0\}}\cb{z\in\R^{m}\mid w^{\mathsf{T}}z\geq w^{\mathsf{T}}\delta_{\ell,r}(X)+\inf_{x \in C}{{w}^{\mathsf{T}}(r\cdot x)}}.
\end{equation}
\end{lem}

\begin{proof}
With \eqref{set-dual-halfspace} in view, for $r, w \in \R^m_+\bs\cb{0}$, we define
\begin{equation}\label{Mhalfspace}
M(r, w) \coloneqq \cb{\eta \in \R^m \mid w^{\mathsf{T}}\eta \geq w^{\mathsf{T}}\delta_{\ell,r}(X)+\inf_{x\in C}{w^{\mathsf{T}}(r\cdot x)}},
\end{equation}
and we will show
\begin{equation}\label{H=M}
\bigcap_{\substack{\lambda\in\R^m_+,\\v \in \R^m_+\bs\cb{0}}}H(\lambda, v) = \bigcap_{\substack{r\in\R^m_+,\\ w \in \R^m_+\bs\cb{0}}}M(r, w).
\end{equation}
First, if $r\in\R^m_+, w \in \R^m_+\bs\cb{0}$, then we define $\lambda_i = r_iw_i$ and $v_i = w_i$ for $i \in \cb{1, \ldots, m}$. Then,
$\lambda\in\R^m_+, v  \in \R^m_+\bs\cb{0}$ as well as $H(\lambda, v) = M(r, w)$; see \eqref{vectordrm} and \eqref{EqDualDiv1}. This means that the intersection on the left hand side runs over at least as many sets as the one on the right hand side; hence, ``$\subseteq$" holds true. Conversely, if $\lambda\in\R^m_+, v \in \R^m_+\bs\cb{0}$, then we define, for each $n \in \mathbb{N}$ and $i \in \cb{1, \ldots, m}$,
\begin{equation}
(r_i^n,w_i^n)\coloneqq\begin{cases}\of{\frac{\lambda_i}{v_i},v_i }&\text{if } v_i>0,\\ \of{1,v_i }&\text{if }v_i =0, \lambda_i = 0,\\ \of{n \lambda_i,\frac{1}{n} }&\text{if } v_i = 0, \lambda_i > 0.\end{cases}
\end{equation}
Then $r^n\in\R^m_+, w^n \in \R^m_+\bs\cb{0}$ and $\lambda_i = r^n_iw^n_i$. Let $\eta$ be a point in the right hand side of \eqref{H=M}. If there is no $i \in \cb{1, \ldots, m}$ satisfying $v_i = 0$ and $\lambda_i > 0$, then $v = w^n$ and $H(\lambda, v) = M(r^n, w^n)$ for every $n\in\mathbb{N}$; hence, $\eta \in H(\lambda, v)$. Next, assume there is some $j \in \cb{1, \ldots, m}$ with $v_j = 0$, $\lambda_j > 0$. Since $\eta \in M(r^n, w^n)$ for every $n \in\mathbb{N}$, it follows
\begin{align}\label{EqMIneqality}
&\sum_{\substack{i \colon v_i > 0 \\ i \colon v_i=\lambda_i=0}}\negthickspace\negthickspace v_{i}\eta_{i} +
	\sum_{i \colon v_i = 0, \, \lambda_i > 0}\negthickspace \frac{\eta_i}{n} \notag \\
&\quad \geq 
 \negthickspace\sum_{\substack{i \colon v_i > 0 \\ i \colon v_i=\lambda_i=0}}\negthickspace\negthickspace 
		\inf_{z_{i} \in \R\colon\E\sqb{\ell(-X_i-z_i)}<+\infty}\of{v_{i}z_{i} + \lambda_{i}\E\sqb{\ell_{i}(-X_i - z_i)}} \notag  \\
	& \quad \quad \quad + \negthickspace\sum_{i \colon v_i = 0, \, \lambda_i > 0}\negthickspace\frac{1}{n}\inf_{z_i \in \R}\of{z_i + 
		n\lambda_i\E\sqb{\ell_i(-X_i-z_i)}} + \inf_{x \in  C}{{\lambda}^{\mathsf{T}}x}.
\end{align}
If $j\in\cb{1,\ldots,m}$ such that $v_i=0,\lambda_j > 0$, then we obtain, for each $n \in \mathbb{N}$,
\begin{align}
-\esssup{X_{j}} - n\lambda_jg_j\of{\frac{1}{n\lambda_j}}  &\leq \inf_{z_j \in \R}\of{z_j + n\lambda_j\E\sqb{\ell_j(-X_j - z_j)}}\\
&\leq -\essinf{X_{j}} - n\lambda_jg_j\of{\frac{1}{n\lambda_j}}.\notag 
\end{align}
This can be checked by a similar calculation to the one in \eqref{boundednessinequality}. Since $g_j$ is convex and lower semicontinuous, the restriction of $g_j$ to $\cl\dom g_j$ is a continuous function, see \citet[Proposition~2.1.6]{zalinescu}, so that
\begin{equation}
\lim_{n \to \infty}\frac{1}{n}\inf_{z_j\in\R}\of{z_j + n\lambda_j\E\sqb{\ell_j(-X_j - z_j}} = 
	-\lim_{n \to \infty}\lambda_jg_j\of{\frac{1}{n\lambda_j}} = -\lambda_jg_j(0).
\end{equation}
On the other hand,
\begin{align}\label{EqGofZero}
\inf_{z_j \in \R}\of{v_jz_j+ \lambda_i\E\sqb{\ell_i(-X_j - z_j)}} &=\lambda_j\inf_{z_j \in \R} \E\sqb{\ell_j(-X_j - z_j)} \\
&= \lambda_j\inf_{y \in \R} \ell_j(y) = -\lambda_j g_j(0)\notag 
\end{align}
since $\ell_j$ is nondecreasing and $X_{j} \in L^\infty$. Taking the limit in \eqref{EqMIneqality} as $n \to \infty$, we finally obtain
\begin{align}
v^{\mathsf{T}}\eta & \geq \sum_{i=1}^m\inf_{z_i \in \R\colon\E\sqb{\ell_i(-X_i-z_i)}<+\infty}(v_iz_i + \lambda_i\E\sqb{\ell_i(-X_i-z_i)})+\negthickspace\inf_{x \in C}\lambda^{\mathsf{T}}x,
\end{align}
that is, $\eta\in H(\lambda, v)$. Hence, \eqref{H=M} follows.
\end{proof}

\begin{proof}[Proof of Theorem~\ref{srmdrmconnection}]
Let $r\in\R^m_+$ and define
\begin{equation}
g_{w,r}(z)\coloneqq w^{\mathsf{T}}\of{-z+ r\cdot\E\sqb{\ell(-X+z)}}
\end{equation}
for each $w\in\R^m_+\bs\cb{0}$ and $z\in\R^m$ with $\E\sqb{\ell(-X-z)}\in\R^m$. Note that
\begin{align}
D_{\ell,r}(X)&=\bigcap_{w\in\R^m_+\bs\cb{0}}\cb{ \eta\in\R^m\mid   w^{\mathsf{T}}\eta\geq  \inf_{\substack{z\in\R^m\colon\\\E\sqb{\ell(-X-z)}\in\R^m}} g_{w,r}(z)+\inf_{x\in C}w^{\mathsf{T}}(r\cdot x)}\notag\\
&=\bigcap_{w\in\R^m_+\bs\cb{0}}\cb{\eta\in\R^{m}\mid w^{\mathsf{T}}\eta\geq w^{\mathsf{T}}\delta_{\ell,r}(X)+\inf_{x \in C}{{w}^{\mathsf{T}}(r\cdot x)}},
\end{align}
which follows from Definition~\ref{div-defn}, \eqref{vectordrm}, \eqref{EqDualDiv1}, and the fact that a closed convex set is the intersection of all of its supporting halfspaces; see \citet[(5.2)]{hl:lagrange}. By Theorem~\ref{oce-duality}, we have $\delta_{\ell,r}(X)\in\R^{m}$ if and only if $1\in\dom g_r$. Hence, $D_{\ell,r}(X)=\R^m$ if and only if $1\notin \dom g_r$. The result follows directly from Lemma~\ref{changeofvariable}.
\end{proof}

\begin{proof}[Proof of Proposition~\ref{setvaluedoce}]
If $1\in\dom g_r$, then $\delta_{\ell,r}(X)\in\R^m$ and the computation in the proof of Theorem~\ref{srmdrmconnection} can be concluded as
\begin{align}
D_{\ell,r}(X)&=\delta_{\ell,r}(X)+\bigcap_{w\in\R^{m}_{+}\bs\{0\}}\left\{z\in\R^{m}\mid w^{\mathsf{T}}z\geq \inf_{x\in C}{w}^{\mathsf{T}}(r\cdot x)\right\}\\
&=\delta_{\ell,r}(X)+r\cdot C.\notag 
\end{align}
With this representation, it is easy to check that $D_{\ell,r}$ is a closed convex risk measure since $\delta_{\ell_i, r_i}$ is a lower semicontinuous convex scalar risk measure for each $i\in\cb{1,\ldots,m}$. If $1\notin \dom g_r$, then $\delta_{\ell,r}(X)=-\infty$ and hence $D_{\ell,r}(X)=\R^m$ due to the convention $w^{\mathsf{T}}(-\infty)=-\infty$ in Lemma~\ref{changeofvariable}.
\end{proof}

\begin{proof}[Proof of Proposition~\ref{set-drm-penalty}]
Let $w \in \R^m_+\bs\{0\}$ and $M(r,w)$ as in \eqref{Mhalfspace}. For the moment, let us denote by $-\alpha$ the function defined by the right hand side of \eqref{setdiver}. Using the dual representation of scalar divergence risk measures provided by Theorem~\ref{oce-duality}, we have
\begin{align}
M(r,w)
&\negthinspace =\negthinspace
 \cb{z\in\R^m \mid w^{\mathsf{T}} z \negthinspace\geq\negthinspace \sum_{i=1}^m\sup_{\Q_i \in \M(\Pr)}
	w_i(\E^{\Q_i}\sqb{-X_i} - I_{g_i,r_i}(\Q_i \mid \Pr))
	+ \inf_{x\in C}  w^{\mathsf{T}}(r\cdot x)}\notag \\
& = \bigcap_{\Q\in\M_m(\Pr)}\cb{z \in \R^m \mid w^{\mathsf{T}}z\geq w^{\mathsf{T}}(\E^{\Q}\sqb{-	X}-I_{g,r}(\Q\mid \Pr)) + \inf_{x \in C}  w^{\mathsf{T}}(r\cdot x)}\notag\\
& = \bigcap_{\Q \in\M_m\of{\Pr}}\of{-\alpha(\Q,w)+\E^{\Q}\sqb{-X}}.
\end{align}
Hence, 
\begin{align}
D_{\ell,r}(X) &= \bigcap_{w \in \R^m_+\bs\{0\}}M(r,w)\\ 
&= \bigcap_{(\Q,w) \in \M_m(\Pr)\times(\R^m_+\bs\{0\})}\of{-\alpha(\mathbb{Q},w)+\E^{\Q}\sqb{-X}}.\notag 
\end{align}
Finally, we show that $-\alpha=-\alpha_{D_{\ell,r}}$. Using \eqref{minpen} for $\Q\in\M_{m}(\Pr)$, $w\in\R^{m}_{+}\bs\{0\}$, we obtain
\begin{align}
-\alpha_{D_{\ell,r}}(\Q,w)
&=\cl\bigcup_{X\in L_{m}^{\infty}}\of{\E^{\Q}\sqb{X}+\delta_{\ell,r}(X)+r\cdot C+G(w)}\\
&=\cb{z\in\R^{m}\mid w^{\mathsf{T}}z\geq \inf_{X\in L_{m}^{\infty}}w^{\mathsf{T}}\of{\E^{\Q}\sqb{X}+\delta_{\ell,r}(X)}+\inf_{x\in C}w^{\mathsf{T}}(r\cdot x)}\notag\\
&=\left\{z\in\R^{m}\mid w^{\mathsf{T}}z\geq -\sum_{i=1}^{m}w_{i}I_{g_{i},r_{i}}(\Q_{i}\mid\Pr)+\inf_{x\in C}w^{\mathsf{T}}(r\cdot x)\right\}\notag\\
&=-\alpha(\Q,w),\notag 
\end{align}
where the third equality follows from the analogous scalar result established in Proposition~\ref{divergence-minimal}. Hence, $-\alpha_{D_{\ell,r}}=-\alpha$ and \eqref{setdiver} holds.
\end{proof}

\begin{proof}[Proof of Proposition~\ref{set-srm-penaltyname}]
Using \eqref{minpen} for $\Q\in \M_m(\Pr)$, $w \in\R^m_{++}$, we obtain
\begin{align}
&-\alpha_{R_{\ell}}(\Q,w)\notag \\
& = \cl\bigcup_{X \in L_{m}^\infty}\of{\negthinspace\E^{\Q}\sqb{X} + G(w)+ 
	\cl\bigcup_{\substack{z \in \R^m\colon\\ \E\sqb{\ell(-X-z)}\in\R^m}}\negthinspace\negthinspace\cb{z + \R^m_+ \mid 0 \in \E\sqb{\ell(-X-z)}+C}\negthinspace\negthinspace}\notag\\
&=\cl\bigcup_{z \in \R^m} \bigcup_{\substack{X\in L_m^\infty\colon\\\E\sqb{\ell(-X-z)}\in\R^m }}\cb{z +\E^{\Q}\sqb{X}+G(w)
	\mid 0 \in \E\sqb{\ell(-X-z)} + C} \notag\\
& = \cl\bigcup_{X\in L_m^\infty}\cb{\E^{\Q}\sqb{-X}+G(w) \mid 0\in\E\sqb{\ell(X)}+C}\notag \\
&= \inf_{(\G_m, \supseteq)}\cb{\E^{\Q}\sqb{-X}+G(w) \mid 0 \in \E\sqb{\ell(X)}+C, X \in L_m^\infty},
\end{align}
where $\E\sqb{\ell(X)}+C$ is understood to be $\emptyset$ whenever $\E\sqb{\ell(X)}=+\infty$. Next, we compute the optimal value of the dual problem for this convex set-valued minimization problem. By \eqref{lagrange}, for $X \in L_m^\infty$, $\lambda \in \R^m_+$, $v \in \R^m_+\bs\{0\}$, we have $L(X,\lambda,v)=\R^m$ if $v \not\in \cb{sw \mid s > 0}$. Moreover, if $v = sw$ for some $s > 0$, then
\begin{align}
&L(X,\lambda, v) \notag \\
& = \E^{\Q}\sqb{-X}+G(sw)+\cb{z \in \R^m \mid sw^{\mathsf{T}}z \geq \lambda^{\mathsf{T}}\mathbb{E}\sqb{\ell(X)}+\inf_{x\in C}\lambda^{\mathsf{T}}x}\notag \\
& = \cb{z \in \R^m \mid sw^{\mathsf{T}}z\geq sw^{\mathsf{T}}\E^{\Q}\sqb{-X}+\lambda^{\mathsf{T}}\E\sqb{\ell(X)}+\inf_{x\in C}\lambda^{\mathsf{T}}x}
\end{align}
whenever $\E\sqb{\ell(X)}\in\R^m$ and $L(X,\lambda,v)=\emptyset$ otherwise. Observe $G(sw) = G(w)$ for every $s > 0$. Hence,
\begin{align}
H(\lambda, sw) &=\negthinspace \cb{z \in \R^m \negthinspace\mid\negthinspace w^{\mathsf{T}}z \negthinspace\geq\negthinspace \inf_{\substack{X \in L_m^{\infty}\colon\\ \E\sqb{\ell(X)}\in\R^m}}\negthinspace\of{w^{\mathsf{T}}\E^{\Q}\sqb{-X}\negthinspace+\negthinspace\frac{1}{s}			\lambda^{\mathsf{T}}\E\sqb{\ell(X)}}\negthinspace+\negthinspace\inf_{x \in C}\frac{1}{s}\lambda^{\mathsf{T}}x}\negthinspace\notag \\
&=\negthinspace H\negthinspace\of{\frac{\lambda}{s}, w}
\end{align}
for $\lambda \in \R^m_+$, $s > 0$. The optimal value of the dual problem is
\begin{equation}
\sup\cb{H\of{\frac{\lambda}{s}, w} \mid s > 0, \; \lambda \in \R^m_+} = \sup\cb{H\of{\lambda, w} \mid \lambda \in \R^m_+}.
\end{equation}
Since $w_i > 0$ for every $i \in \cb{1, \ldots, m}$ by assumption, we have
\begin{align}
H(\lambda,w)=\Big\{z\in\R^{m} \mid w^{\mathsf{T}}z\geq &\inf_{X\in L_{m}^{\infty}}\of{w^{\mathsf{T}}\E^{\Q}\sqb{-X}+w^{\mathsf{T}}\of{r\cdot\E\sqb{\ell(X)}}}\notag\\
&+\inf_{x\in C}w^{\mathsf{T}}(r\cdot x)\Big\},
\end{align}
where $r_i \coloneqq \frac{\lambda_i}{w_i}$, $i\in\cb{1, \ldots, m}$.
Note that
\begin{align}
&\inf_{X\in L_{m}^{\infty}}\of{w^{\mathsf{T}}\E^{\Q}\sqb{-X}+w^{\mathsf{T}}\of{r\cdot\E\sqb{\ell(X)}}}\notag\\
&=\sum_{i=1}^{m}w_{i}\inf_{X_{i}\in L^{\infty}}\of{w_{i}\E\sqb{-\frac{d\Q_{i}}{d\Pr}X_{i}}+r_{i}\E\sqb{\ell_{i}(X_{i})}}\notag\\
&=\sum_{i=1}^{m}w_{i}\E\sqb{\inf_{x_{i}\in \R}\of{-\frac{d\Q_{i}}{d\Pr}x_{i}+r_{i}\ell_{i}(x_{i})}}\notag\\
&=w^{\mathsf{T}}I_{g,r}(\Q\mid\Pr).
\end{align}
Therefore, the optimal value of the dual problem equals the middle term in \eqref{set-srm-penalty}. Note that Slater's condition holds, that is, there exists $\bar{X}\in L_{m}^{\infty}$ such that $(\E\sqb{\ell(\bar{X})}+C)\cap -\R^{m}_{++}\neq\emptyset$. This is immediate from the scalar version as in the proof of Proposition~\ref{divergence-minimal}. Hence, the first equality in \eqref{set-srm-penalty} holds by \citet[Theorem 6.6]{hl:lagrange}. Since $I_{g,r}(\mathbb{Q}\mid\mathbb{P})\not\in\R^{m}$ if $1\notin \dom g_r$, we also have the second equality in \eqref{set-srm-penalty}.
\end{proof}

\subsection{Proofs of the results in Section~\ref{examples}}\label{exampleproof}

\begin{proof}[Proof of Proposition~\ref{pointplusset}]
Using the definitions, we have
\begin{align}
R^{\ent}(X)
&=\cb{\negthinspace z\in\R^{m}\mid \exists c\in C\;\forall i\in\{1,\ldots,m\}\colon  \frac{\E\sqb{e^{\beta_{i}(-X_{i}-z_{i})}}-1}{\beta_{i}}=-c_{i}\negthinspace}\\
&=\cb{z\in\R^{m}\negthinspace\mid \negthinspace\exists c\in C\;\forall i\in\{1,\ldots,m\}\colon\negthinspace z_{i}=\frac{1}{\beta_{i}}\log\frac{\E\sqb{e^{-\beta_{i}X_{i}}}}{1-\beta_{i}c_{i}},\;1>\beta_{i}c_{i}}\notag\\
&=\rho^{\ent}(X)+C^{\ent}.\notag 
\end{align}
\end{proof}

\begin{proof}[Proof of Proposition~\ref{ent-drm}]
For each $i\in\{1,\ldots,m\}$, note that
\begin{align}
\delta_{\ell_{i},r_{i}}(X_{i})&=\inf_{z_{i}\in\R}\of{z_{i}+r_{i}\E\sqb{\ell_{i}(-X_{i}-z_{i})}}\\
&=\frac{1}{\beta_{i}}\log\E\sqb{e^{-\beta_{i}X_{i}}}+\frac{1}{\beta_{i}}(1-r_{i}+\log r_{i})\in\R.\notag 
\end{align}
The result follows from Proposition~\ref{setvaluedoce}.
\end{proof}

\begin{proof}[Proof of Lemma~\ref{bestratio}]
First, we extend $f_{w}$ and $h_{w}$ from $\R^{m}_{++}$ to $\R^{m}$ with their original definitions so that we have $\inf_{r\in\R^{m}_{++}}\of{f_{w}(r)+h_{w}(r)} = \inf_{r\in\R^{m}}\of{f_{w}(r)+h_{w}(r)}$. Note that $f_{w}$ is a proper, strictly convex, continuous function and has a unique minimum point. Hence, by \citet[Theorem 27.1(d)]{rockafellar}, $f_{w}$ has no directions of recession, that is, the recession function $f_{w}0^{+}$ of $f_{w}$ always takes strictly positive values; see \citet[p.~66 and p.~69]{rockafellar} for definitions. Besides, $h_{w}$ is a proper, convex, lower semicontinuous function. If $h_{w}\equiv+\infty$, then the infimum of $f_{w}+h_{w}$ is $+\infty$. Suppose that $h_{w}$ is a proper function. Since $0$ is a boundary point of $-C$, $h_{w}$ always takes nonnegative values. Hence, the infimum of $h_{w}$ is finite. By \citet[Theorem~27.1(a), (i)]{rockafellar}, this implies that the recession function $h_{w}0^{+}$ of $h_{w}$ always takes nonnegative values. Therefore, $f_{w}+h_{w}$ has no directions of recession since $(f_{w}+h_{w})0^{+}=f_{w}0^{+}+h_{w}0^{+}$ by \citet[Theorem~9.3]{rockafellar}. Hence, by \citet[Theorem~27.1(b), (d)]{rockafellar} and the strict convexity of $f_{w}+h_{w}$, this function has a unique minimum point $r^{w}\in\R^{m}_{++}$ which is determined by the first order condition
\begin{align}
0\in&\partial (f_{w}+h_{w})(r^{w})\\
&=\sqb{\frac{w_{i}}{\beta_{i}}-\frac{w_{i}}{\beta_{i}r^{w}_{i}}}_{i=1}^{m}+\cb{w\cdot\bar{x}\mid\bar{x}\in -C,\;\sup_{x\in -C}w^{\mathsf{T}}(r^{w}\cdot x) =w^{\mathsf{T}}(r^{w}\cdot\bar{x})},\notag 
\end{align}
that is,
\begin{equation}\label{exchangeform}
\sqb{\frac{1}{\beta_{i}}\of{1-\frac{1}{r^{w}_{i}}}}_{i=1}^{m}\in C,\quad \inf_{x\in C}w^{\mathsf{T}}(r^{w}\cdot x) =\sum_{i=1}^{m}\frac{w_{i}r^{w}_{i}}{\beta_{i}}\of{1-\frac{1}{r^{w}_{i}}},
\end{equation}
which is the claimed property of $r^{w}$.
\end{proof}

\begin{proof}[Proof of Proposition~\ref{ent-setopt}]
By Lemma~\ref{bestratio} , it is clear that, for each $w \in \R^{m}_{+}\bs\{0\}$, we have $\inf_{r\in\R^{m}_{++}}\of{f_{w}(r)+h_{w}(r)} = \inf_{r\in \Gamma}\of{f_{w}(r)+h_{w}(r)}=f_{w}(r^{w})+h_{w}(r^{w})$. Hence,
\begin{align}
R^{\ent}(X)
&=\rho^{\ent}(X)+\bigcap_{\substack{w\in\R^{m}_{+}\bs\{0\},\\ r\in\R^{m}_{++}}}\cb{z\in\R^{m}\mid w^{\mathsf{T}}z\geq -(f_{w}(r)+h_{w}(r))}\\
&=\rho^{\ent}(X)+\bigcap_{w\in\R^{m}_{+}\bs\{0\}}\cb{z\in\R^{m}\mid w^{\mathsf{T}}z\geq -\inf_{r\in\Gamma}\of{f_{w}(r)+h_{w}(r)}}\notag\\
&=\bigcap_{r\in\Gamma}D_{r}^{\ent}(X).\notag
\end{align}
Let $w\in\R^{m}_{+}\bs\{0\}$ such that $f_{w}+h_{w}$ is proper and let $r\in\R^{m}_{++}$. Suppose that $D_{r}^{\ent}(X)\subseteq D_{r^{w}}^{\ent}(X)$. Then, $-(f_{w}(r)+h_{w}(r))=\inf_{z\in D_{r}^{\ent}(X)}w^{\mathsf{T}}z\geq\inf_{z\in D_{r^{w}}^{\ent}(X)}w^{\mathsf{T}}z=-(f_{w}(r^{w})+h_{w}(r^{w}))$, that is, $f_{w}(r)+h_{w}(r)\leq f_{w}(r^{w})+h_{w}(r^{w})$. By Lemma~\ref{bestratio}, this implies that $r=r^{w}$.
\end{proof}

\begin{proof}[Proof of Proposition~\ref{ent-dual}]
Proposition~\ref{set-srm-penaltyname} and Lemma~\ref{bestratio} give
\begin{align}\label{entropicpenalty}
&-\alpha_{R_{\ell}}(\Q,w)\\
&
=\bigcap_{r\in 1/\dom g}\negthinspace\cb{z\in\R^{m}\mid w^{\mathsf{T}}z\geq  -w^{\mathsf{T}}I_{g,r}(\Q\mid\Pr)\negthinspace+\negthinspace\inf_{x\in C}{w^{\mathsf{T}}(r\cdot x)}}\notag \\
&=\negthinspace-\beta^{-1}\negthinspace\cdot\negthinspace H(\Q\mid\Pr)\negthinspace+\negthinspace\bigcap_{r\in \R^{m}_{++}}\left\{\negthinspace z\in\R^{m}\negthinspace\mid\negthinspace w^{\mathsf{T}}z\geq  \sum_{i=1}^{m}\frac{w_{i}}{\beta_{i}}(1-r_{i}+\log r_{i})\negthinspace+\negthinspace\inf_{x\in C}{w^{\mathsf{T}}(r\negthinspace\cdot\negthinspace x)}\negthinspace\right\}\notag\\
&=-\beta^{-1}\negthinspace\cdot\negthinspace H(\Q\mid\Pr)+\bigcap_{r\in \R^{m}_{++}}\cb{z\in\R^{m}\mid w^{\mathsf{T}}z\geq  -(f_{w}(r)+h_{w}(r))}\notag\\
&=-\beta^{-1}\negthinspace\cdot\negthinspace  H(\Q\mid\Pr)+\cb{z\in\R^{m}\mid w^{\mathsf{T}}z\geq  -(f_{w}(r^{w})+h_{w}(r^{w}))}\notag\\
&=-\beta^{-1}\negthinspace\cdot\negthinspace  H(\Q\mid\Pr)\negthinspace+\negthinspace\left\{z\in\R^{m}\mid w^{\mathsf{T}}z\geq  \sum_{i=1}^{m}\frac{w_{i}}{\beta_{i}}(1-r^{w}_{i}+\log r^{w}_{i})\negthinspace+\negthinspace\inf_{x\in C}{w^{\mathsf{T}}(r^{w}\cdot x)}\negthinspace\right\},\notag
\end{align}
assuming that $h_{w}$ is not identically $+\infty$ (otherwise $-\alpha_{R_{\ell}}(\mathbb{Q},w)=\R^m$). The passage from the last line to the claimed formula is by \eqref{exchangeform}.
\end{proof}

\subsection{Proofs of the results in Section~\ref{market}}\label{marketproof}

\begin{proof}[Proof of Proposition~\ref{markettransition}]
Clearly, $R^{\market}(0)\neq\emptyset$ since $0\in \Lambda_m(0)$ and $R(0)\neq \emptyset$. We prove the monotonicity and translativity of the function $Y\mapsto \tilde{R}(Y)\coloneqq \bigcup_{X\in\Lambda_m(Y)}R(X)$ first. For monotonicity, consider $Y^1, Y^2 \in L_d^\infty$ with $Y^1 \leq Y^2$. Let $X\in \Lambda_m(Y^1)$. With $\tilde{Y}\coloneqq Y^2-Y^1 \in L_{d,+}^\infty$, it holds
\begin{align}
BX  \in Y^1 +\K &= Y^2 -\tilde{Y}+ \K \\
&=Y^2 -\sum_{t=0}^{T-1}L_{d}^{\infty}(\F_{t},\C_{t}\cap\D_t)- \of{\tilde{Y} + L_d^\infty(\F_T, \C_T)} \notag \\
& \subseteq Y^2 -\sum_{t=0}^{T-1}L_{d}^{\infty}(\F_{t},\C_{t}\cap\D_t)- \of{L_{d,+}^\infty + L_d^\infty(\F_T, \C_T)}\notag\\
& \subseteq Y^2 + \K,\notag
\end{align}
where the last inclusion holds since
$L_{d,+}^\infty + L_d^\infty (\F_T, \C_T) = L_d^\infty (\F_T, \R^d_+)+L_d^\infty(\F_T, \C_T)=L_d^\infty (\F_T, \C_T)$
due to \mbox{$\C_T(\o)\in \G_d$} for every $\o\in\O$. Hence, $X\in \Lambda_m (Y^2)$. Therefore, $\Lambda_m (Y^1)\subseteq \Lambda_m (Y^2)$, which implies $\tilde{R}(Y^1)\subseteq \tilde{R}(Y^2)$. To prove translativity, let $Y\in L_d^\infty$, $z\in\R^m$. For every $X\in L_m^\infty$, it holds
\begin{align}
X\in \Lambda_m (Y+Bz) &\; \Leftrightarrow \; BX \in Y+Bz+\K\\
&\;\Leftrightarrow \; B(X-z) \in Y+\K\notag\\
&\;\Leftrightarrow \; X-z \in \Lambda_m (Y).\notag 
\end{align}
Hence,
\begin{align}
\tilde{R}(Y+Bz) &= \bigcup_{X\in\Lambda_m(Y+Bz)}R(X)
= \bigcup_{X-z\in\Lambda_m(Y)}R(X)\\
&= \bigcup_{X\in\Lambda_m(Y)}R(X+z)
=\tilde{R}(Y)-z,\notag 
\end{align}
from which translativity follows. It is easy to check that the last two properties are preserved under the closure and convex hull operators. Hence, $R^{\market}$ is monotone and translative. It is also easy to check that $\tilde{R}$ and $R^{\market}$ are convex since $R$ is convex. Finally, since $\tilde{R}$ has convex values and this property is preserved under the closure operator, \eqref{co-dropped} follows.
\end{proof}

As a preparation for the proof of Theorem~\ref{liquidationtheorem}, we establish a link between the notions of market risk measure and set-valued infimal convolution. We begin by introducing two key concepts from (complete lattice-based) set-valued convex analysis, the reader is referred to \cite{andreasduality} for details.

\begin{defn}\label{indicator}
(\citealt[Example~1]{andreasduality}) Let $\mathcal{Y}\subseteq L_{d}^{\infty}$. The \emph{indicator function} of the set $\mathcal{Y}$ is the function $\mathcal{I}^{m}_{\mathcal{Y}} \colon L_{d}^{\infty}\to\G_m$ defined by
\begin{equation}
\mathcal{I}^{m}_{\mathcal{Y}}(Y)=\begin{cases}\R^{m}_{+}& \text{if }Y\in\mathcal{Y},\\ \emptyset & \text{else}.\end{cases}
\end{equation}
\end{defn}

\begin{defn}\label{infimalconvolution}
(\citealt[Section~4.4(C)]{andreasduality}) Let $N\geq 1$ be an integer. For each $n\in\{1, \ldots, N\}$, let $F^{n} \colon L_{d}^{\infty}\to\G_m$ be a function. The function $ \Box_{n=1}^{N}F^{n} \colon L_{d}^{\infty}\to\G_m$ defined by
\begin{equation}
(\Box_{n=1}^{N}F^{n})(Y)=\cl\co\bigcup_{Y^{1}, \ldots,Y^{N}\in L_{d}^{\infty}}\cb{\sum_{n=1}^{N}F^{n}(Y^{n})\mid Y^{1}+\ldots+Y^{N}=Y}.
\end{equation}
is called the \emph{infimal convolution} of $F^{1}, \ldots, F^{N}$.
\end{defn}

Recall the linear operator $B \colon \R^{m}\rightarrow\R^{d}$ defined by \eqref{linop}: $Bx=(x_{1}, \ldots, x_{m},0, \ldots, 0)^{\mathsf{T}}$ for $x\in\R^{m}$. Its adjoint $B^{*} \colon \R^{d}\rightarrow\R^{m}$ is defined by \eqref{adjop}: $B^{*}y=(y_{1}, \ldots ,y_{m})^{\mathsf{T}}$ for $y\in\R^{d}$.

The next lemma shows that the market risk measure is basically the infimal convolution of the original risk measure and the indicator function of the negative of the set $\mathcal{K}$ of all freely available portfolios defined by \eqref{freelyavailable}.

\begin{lem}\label{infconvmarket}
Let $R \colon L_{m}^{\infty}\to\G_m$ be a closed convex risk measure and define $\tilde{R} \colon L_{d}^{\infty}\to\G_m$ by 
\begin{equation}
\tilde{R}(Y)=\begin{cases}R(B^{*}Y)&\text{ if }Y\in B(L_{m}^{\infty}),\\\emptyset&\text{ else}.\end{cases}
\end{equation}
Then, for each $Y\in L_{d}^{\infty}$,
\begin{align}\label{closedvalues}
R^{\market}(Y)&=(\tilde{R}\;\Box\;\mathcal{I}^{m}_{-\mathcal{K}})(Y)\\
&=(\tilde{R}\;\Box\;\mathcal{I}^{m}_{L_{d}^{\infty}(\mathcal{F}_{0},\mathcal{C}_{0}\cap\mathcal{D}_{0})}\;\Box\;\ldots\;\Box\;\mathcal{I}^{m}_{L_{d}^{\infty}(\mathcal{F}_{T},\mathcal{C}_{T}\cap\mathcal{D}_{T})})(Y).\notag 
\end{align}
\end{lem}

\begin{proof}
For each $Y\in L_{d}^{\infty}$, we have
\begin{align}
&R^{\market}(Y)=\cl\bigcup_{\{X\in L_{m}^{\infty}\mid BX\in Y+\mathcal{K}\}}R(X)=\cl\bigcup_{U\in Y+\mathcal{K}}\tilde{R}(U) \\
&= \cl \bigcup_{U,U^{\prime} \in L_{d}^{\infty}}\{\tilde{R}(U)+\mathcal{I}^{m}_{-\mathcal{K}}(U^{\prime})\mid U+U^{\prime}=Y\}\notag\\
&=\cl \bigcup_{U,U^{0},\ldots,U^{T}\in L_{d}^{\infty}}\left\{\tilde{R}(U)+\sum_{t=0}^{T}\mathcal{I}^{m}_{L_{d}^{\infty}(\mathcal{F}_{t},\mathcal{C}_{t}\cap\mathcal{D}_{t})}(U^{t})\mid U+U^{0}+\ldots+U^{T}=Y\right\}.\notag 
\end{align}
Since each of the functions in the infimal convolution is convex, we can omit the convex hull operator in Definition~\ref{infimalconvolution}; and the result follows.
\end{proof}

By Lemma~\ref{infconvmarket}, the market risk measure can be formulated as an infimal convolution. As in the scalar theory, the Legendre-Fenchel conjugate of the infimal convolution of finitely many convex functions is the sum of the Legendre-Fenchel conjugates of these convex functions; see \citet[Lemma~2]{andreasduality}. The application of this result is the main step of the proof of Theorem~\ref{liquidationtheorem} below. For completeness, we begin with the definition of conjugate for set-valued functions.

\begin{defn}\label{conjugate}
(\citealt[Definition~5]{andreasduality}) Let $F:L_{d}^{\infty}\to\G_m$ be a function. The \emph{(Fenchel) conjugate} of $F$ is the function $-F^{*}:L_{d}^{1}\times (\R^{m}_{+}\bs\{0\})$ defined by
\begin{equation}
-F^{*}(V,v)=\cl\bigcup_{Y\in L_{d}^{\infty}}\of{F(Y)+\cb{z\in \R^{m}\mid v^{\mathsf{T}}z\geq \E\sqb{-V^{\mathsf{T}}Y}}}.
\end{equation}
\end{defn}

\begin{proof}[Proof of Theorem~\ref{liquidationtheorem}]
Since $\cl R^{\market}$ has closed values, Lemma~\ref{infconvmarket} implies that, for each $Y\in L_{d}^{\infty}$,
\begin{equation}
 (\tilde{R}\;\Box\;\mathcal{I}^{m}_{L_{d}^{\infty}(\mathcal{F}_{0},\mathcal{C}_{0}\cap\mathcal{D}_{0})} \;\Box\ldots\Box\;\mathcal{I}^{m}_{L_{d}^{\infty}(\mathcal{F}_{T},\mathcal{C}_{T}\cap\mathcal{D}_{T})})(Y)\negthinspace =\negthinspace R^{\market}(Y)\subseteq (\cl R^{\market})(Y).
\end{equation}
By \citet[Remark~6, Lemma~2]{andreasduality}, $R^{\market}$ and $\cl R^{\market}$ have the same conjugate on $L_{d}^{1}\times(\R^{m}_{+}\bs\{0\})$ given~by
\begin{align}\label{conjugation}
&-\left(\tilde{R}\;\Box\;\mathcal{I}^{m}_{L_{d}^{\infty}(\mathcal{F}_{0},\mathcal{C}_{0}\cap\mathcal{D}_{0})}\;\Box\;\ldots\;\Box\;\mathcal{I}^{m}_{L_{d}^{\infty}(\mathcal{F}_{T}, \mathcal{C}_{T}\cap\mathcal{D}_{T})}\right)^{*}\notag\\
&=-\tilde{R}^{*}+\sum_{t=0}^{T}-(\mathcal{I}^{m}_{L_{d}^{\infty}(\mathcal{F}_{t}, \mathcal{C}_{t}\cap\mathcal{D}_{t})})^{*}.
\end{align}
Note that this is the set-valued version of the rule ``the conjugate of the infimal convolution of finitely many convex functions is the sum of their conjugates." Let $(V,v)\in L_{d}^{1}\times (\R^{m}_{+}\bs\{0\})$. By \citet[Proposition~6.7]{hhr:setval} on the conjugate of a risk measure, for every $(V,v)\in L_{d}^{1}\times(\R^{m}_{+}\bs\{0\})$, we have $-(\cl(R^{\market}(\cdot)))^{*}(V,v)=\R^{m}$ unless we have $V\in -L_{d,+}^{1}$ and $v=\E\sqb{-B^{*}V}$.

Next, we pass from $L_{d}^{1}\times(\R^{m}_{+}\bs\{0\})$ to $\mathcal{W}_{m,d}=\M_{d}(\Pr)\times((\R^{m}_{+}\bs\{0\})\times\R^{d-m}_{+}) $ using the ``change of variables formula" \cite[Lemma~3.4]{hhr:setval}. One obtains that for every $V\in -L_{d,+}^{1}$ with $v=\E\sqb{-B^{*}V}$, there exists $(\Q,w)\in\mathcal{W}_{m,d}$ such that, for every $Y\in L_{d}^{\infty}$,
\begin{equation}\label{changeofvar}
\cb{z\in\R^{m}\mid v^{\mathsf{T}}z\geq \mathbb{E}\sqb{(-V)^{\mathsf{T}}Y}}=B^{*}\left(\left(\mathbb{E}^{\mathbb{Q}}[Y]+G(w)\right)\cap B(\R^{m})\right),
\end{equation}
and conversely, every $(\mathbb{Q},w)\in\mathcal{W}_{m,d}$ can be obtained by some $V\in -L_{d,+}^{1}$ with $v=\mathbb{E}[-B^{*}V]$ such that \eqref{changeofvar} holds for every $Y\in L_{d}^{\infty}$. Note that $B^{*}(\R^{d})=\R^{m}\times\{0\in\R^{d-m}\}$. For such corresponding pairs $(V,v)$ and $(\mathbb{Q},w)$, using \eqref{changeofvar}, we first observe that
\begin{align}
-\tilde{R}^{*}(V,v)&=\cl\bigcup_{Y\in L_{d}^{\infty}}\of{\tilde{R}(Y)+\cb{z\in \R^{m}\mid v^{\mathsf{T}}z\geq \E\sqb{-V^{\mathsf{T}}Y}}}\\
&=\cl\bigcup_{Y\in L_{d}^{\infty}}\of{\tilde{R}(Y)+B^{*}\of{\of{\E^{\Q}[Y]+G(w)}\cap B(\R^{m})}}\notag \\
&=\cl\bigcup_{Y\in B(L_{m}^{\infty})}\of{R(B^{*}Y)+\E^{{B^{*}\Q}}[B^{*}Y]+G(B^{*}w)}\notag\\
&=\cl\bigcup_{X\in L_{m}^{\infty}}\of{R(X)+\E^{{B^{*}\Q}}[X]+G(B^{*}w)}=-\alpha_{R}(B^{*}\Q,B^{*}w).\notag 
\end{align}
Next, let $t\in\T$. For the same pairs $(V,v)$ and $(\Q,w)$, by Definitions~\ref{indicator},~\ref{conjugate}, we have
\begin{align}
\negthinspace-\negthinspace \of{\mathcal{I}^{m}_{L_{d}^{\infty}(\mathcal{F}_{t}, \mathcal{C}_{t}\cap\mathcal{D}_{t})}}^{*}\negthinspace (V,v)\negthinspace &=\negthinspace\cl\bigcup_{U^{t}\in L_{d}^{\infty}(\mathcal{F}_{t}, \mathcal{C}_{t}\cap\mathcal{D}_{t})}\cb{\negthinspace z\in\R^{m}\negthinspace \mid \negthinspace v^{\mathsf{T}}z\negthinspace \geq\negthinspace \E\sqb{(-V)^{\mathsf{T}}U^{t}}\negthinspace}\\
&=\cl\bigcup_{U^{t}\in L_{d}^{\infty}(\mathcal{F}_{t}, \mathcal{C}_{t}\cap\mathcal{D}_{t})}B^{*}\of{\of{\E^{\Q}[U^{t}]+G(w)}\cap B(\R^{m})}.\notag 
\end{align}

Finally, note that $\cl R^{\market}$ is a closed convex set-valued function that is finite at zero by assumption. Hence, by biconjugation for set-valued functions, see \cite[Theorem~2]{andreasduality}, we have
\begin{equation}
(\cl R^{\market})(Y)\negthinspace=\negthinspace\bigcap_{\substack{V\in -L_{d,+}^{1},\\ v=\E\sqb{-B^{*}V}}}\sqb{\negthinspace-(\cl R^{\market})^{*}(V,v)\negthinspace+\negthinspace \cb{z\in\R^{m}\mid v^{\mathsf{T}}z\geq \E\sqb{V^{\mathsf{T}}Y}}\negthinspace},
\end{equation}
for every $Y\in L_{d}^{\infty}$, and the above calculations allow for a passage to vector probability measures:
\begin{equation}
\negthinspace(\cl\negthinspace R^{\market})(Y)\negthinspace=\negthinspace\bigcap_{(\Q,w)\in \mathcal{W}_{m,d}\negthinspace}\negthinspace\sqb{\negthinspace-\alpha_{\cl R^{\market}}(\Q,w)\negthinspace+\negthinspace B^{*}\negthinspace\of{(\E^{\Q}[-Y]\negthinspace+\negthinspace G(w))\negthinspace\cap\negthinspace B(\R^{m})}\negthinspace}\negthinspace,\negthinspace
\end{equation}
where, for $(\Q,w)\in\mathcal{W}_{m,d}$,
\begin{align}
-\alpha_{\cl R^{\market}}(\Q,w)\negthinspace=\negthinspace&-\alpha_{R}(B^{*}\Q,B^{*}w)\\
&+\sum_{t=0}^{T}\cl\bigcup_{U^{t}\in L_{d}^{\infty}(\mathcal{F}_{t}, \mathcal{C}_{t}\cap \mathcal{D}_{t})}B^{*}\of{\of{\E^{\Q}[U^{t}]+G(w)}\cap B(\R^{m})}.\notag
\end{align}
\end{proof}

\begin{proof}[Proof of Corollary~\ref{corollary_convex market}]
Let $(\Q,w)\in\mathcal{W}_{m,d}\bs\mathcal{W}_{m,d}^{\convex}$. So there exist $t\in\T$ and $A\in\mathcal{F}_{t}$ such that $\Pr(A)>0$ and $w\cdot\E\sqb{\frac{d\Q}{d\Pr}\mid\mathcal{F}_{t}}(\omega)\notin (0^{+}\mathcal{C}_{t}(\omega))^{+}$ for each $\omega\in A$. Using the fact that the effective domain of the support function of a nonempty closed convex set in $\R^{d}$ is a subset of its recession cone, which is an easy consequence of \citet[Corollary~14.2.1]{rockafellar}, we see that
$\inf_{y^{t}\in \mathcal{C}_{t}(\omega)}\of{w\cdot \E\sqb{\frac{d\Q}{d\Pr}\;\middle\vert\;\mathcal{F}_{t}}(\omega)}^{\mathsf{T}}y^{t}=-\infty$ for each $\omega\in A$. Note that
\begin{align}
&\cl\bigcup_{U^{t}\in L_{d}^{\infty}(\mathcal{F}_{t}, \mathcal{C}_{t})}  B^{*}\of{\of{\E^{\Q}[U^{t}]+G(w)}\cap B(\R^{m})}\notag \\
&=\cb{z\in \R^{m}\mid w^{\mathsf{T}}(Bz)\geq \inf_{U^{t}\in L_{d} ^{\infty}(\mathcal{F}_{t},\mathcal{C}_{t})}w^{\mathsf{T}}\E^{\Q}\sqb{U^{t}}}\notag\\
&=\cb{z\in \R^{m}\mid (B^{*}w)^{\mathsf{T}}z\geq \E\sqb{\inf_{y^{t}\in \mathcal{C}_{t}}\of{w\cdot\E\sqb{\frac{d\Q}{d\Pr}\;\middle\vert\; \mathcal{F}_{t}}}^{\mathsf{T}}y^{t}}},
\end{align}
where the last equality is by \citet[Theorem~14.60]{rockafellar2}. Note that the passage to conditional expectations in the third line is necessary for the application of this theorem. Since $\Pr(A)>0$, this implies $\cl\bigcup_{U^{t}\in L_{d}^{\infty}(\mathcal{F}_{t},\mathcal{C}_{t})}B^{*}((\E^{\Q}\sqb{U^{t}}+G(w))\cap B(\R^{m}))=\R^{m}$. By the computation in the proof of Proposition~\ref{liquidationtheorem}, it follows that $-\alpha_{\cl R^{\market}}(\Q,w)=\R^{m}$.
\end{proof}

\begin{proof}[Proof of Corollary~\ref{liquidationcorollary}]
Let $t\in\T$. For each $\o\in\O$, we have
\begin{equation}
\negthinspace\inf_{y^{t}\in \mathcal{C}_{t}(\omega)}\of{w\cdot\E\sqb{\frac{d\Q}{d\Pr}\middle\vert\mathcal{F}_{t}}\negthinspace(\omega)}^{\mathsf{T}}y^{t}
=\begin{cases}0 & \text{if }w\negthinspace\cdot\negthinspace\E\sqb{\frac{d\Q}{d\Pr}\middle\vert\mathcal{F}_{t}}(\omega)\in (\mathcal{C}_{t}(\o))^{+}\negthinspace,\negthinspace\\ -\infty &\text{else}\end{cases}
\end{equation}
since $\mathcal{C}_{t}(\o)$ is a nonempty closed convex cone. Similar to the calculation in the proof of Corollary~\ref{corollary_convex market}, we have
\begin{align}
&\cl\bigcup_{U^{t}\in L_{d}^{\infty}(\mathcal{F}_{t},\mathcal{C}_{t})}  B^{*}\of{\E^{\Q}\sqb{U^{t}}+G(w))\cap B(\R^{m})}\notag\\
&=\cb{z\in \R^{m}\mid (B^{*}w)^{\mathsf{T}}z\geq \E\sqb{\inf_{y^{t}\in \mathcal{C}_{t}}\of{w\cdot\E\sqb{\frac{d\Q}{d\Pr} \;\middle\vert\;\mathcal{F}_{t}}}^{\mathsf{T}}y^{t}}},
\end{align}
from which the result follows immediately.
\end{proof}

\begin{proof}[Proof of Proposition~\ref{finitedivergencerm}]
For the first part, let $i\in\{1,\ldots,m\}$. From Remark~\ref{divergenceindex}, recall that $r_{i}\ell_{i}(s)=\sup_{y\in\R}\of{sy-r_{i}g_{i}\of{\frac{y}{r_{i}}}}\geq s-r_{i}g_{i}\of{\frac{1}{r_{i}}}$ for every $s\in\R$. Hence, given $X\in L_{m}^{\infty}$,
\begin{equation}\label{ocebound}
\delta_{\ell_{i},r_{i}}(X_{i})=\inf_{y\in\R}\of{y+r_{i}\E\sqb{\ell_{i}(-X_{i}-y)}} \geq -\E\sqb{X_{i}}-r_{i}g_{i}\of{\frac{1}{r_{i}}}
\end{equation}
for every $i\in\cb{1,\ldots,m}$. Then,
\begin{align}
&\inf_{z\in D_{\ell,r}^{\market}(0)}\of{B^{*}\bar{w}}^{\mathsf{T}}z\\
&= \inf_{X \in \Lambda_m(0)}\inf_{z\in D_{\ell,r}(X)}\of{B^{*}\bar{w}}^{\mathsf{T}}z\notag \\
&=\inf_{X\in \Lambda_m(0)}\of{B^{*}\bar{w}}^{\mathsf{T}}\delta_{\ell,r}(X)+\inf_{x\in C}\of{B^{*}\bar{w}}^{\mathsf{T}}\of{r\cdot x}\notag\\
&\geq \inf_{X\in \Lambda_m(0)}\of{B^{*}\bar{w}}^{\mathsf{T}}\E\sqb{-X}-\sum_{i=1}^{m}\bar{w}_{i}r_{i}g_{i}\of{\frac{1}{r_{i}}}+\inf_{x\in C}\of{B^{*}\bar{w}}^{\mathsf{T}}\of{r \cdot x}\notag\\
&= \inf_{X\in \Lambda_m(0)}\bar{w}^{\mathsf{T}}\E\sqb{-BX}-\sum_{i=1}^{m}\bar{w}_{i}r_{i}g_{i}\of{\frac{1}{r_{i}}}+\inf_{x\in C}\of{B^{*}\bar{w}}^{\mathsf{T}}\of{r \cdot x}\notag\\
&\geq \inf_{Y\in \K}\bar{w}^{\mathsf{T}}\E\sqb{-Y}-\sum_{i=1}^{m}\bar{w}_{i}r_{i}g_{i}\of{\frac{1}{r_{i}}}+\inf_{x\in C}\of{B^{*}\bar{w}}^{\mathsf{T}}\of{r \cdot x}\notag\\
&=\sum_{t=0}^{T}\inf_{U\in L_{d}^{\infty}(\F_{t},\mathcal{C}_{t}\cap\mathcal{D}_t)}\E\sqb{\bar{w}^{\mathsf{T}}U}-\sum_{i=1}^{m}\bar{w}_{i}r_{i}g_{i}\of{\frac{1}{r_{i}}}+\inf_{x\in C}\of{B^{*}\bar{w}}^{\mathsf{T}}\of{r\cdot x}\notag\\
&\geq \sum_{t=0}^{T}\inf_{U\in L_{d}^{\infty}(\mathcal{F}_{t},\mathcal{C}_{t})}\E\sqb{\bar{w}^{\mathsf{T}}U}-\sum_{i=1}^{m}\bar{w}_{i}r_{i}g_{i}\of{\frac{1}{r_{i}}}+\inf_{x\in C}\of{B^{*}\bar{w}}^{\mathsf{T}}\of{r\cdot x}\eqqcolon a,\notag 
\end{align}
where the first inequality follows from \eqref{ocebound}, the second inequality follows since $\Lambda_{m}(0)$ $= \{X\in L_{m}^{\infty}$ $ \mid BX\in \K\}$, and the last inequality follows since $L_d^{\infty}(\F_t,\C_t\cap\D_t)\subseteq L_d^{\infty}(\F_t,\C_t)$ for each $t\in\cb{0,\ldots,T}$. By the same arguments as in the proofs of Corollary~\ref{corollary_convex market} and Corollary~\ref{liquidationcorollary}, the hypotheses guarantee that $a > -\infty$. Hence,
\begin{align}
D_{\ell,r}^{\market}(0) &\subseteq \cb{\eta\in\R^{m}\mid\of{B^{*}\bar{w}}^{\mathsf{T}}\eta\geq \inf_{z\in D_{\ell,r}^{\market}(0)}\of{B^{*}\bar{w}}^{\mathsf{T}}z}\notag \\& \subseteq \cb{\eta\in\R^m\mid \of{B^{*}\bar{w}}^{\mathsf{T}}\eta\geq a}\negthinspace\neq \negthinspace\R^m.
\end{align}
Note that $L_{m}^{\infty}\ni X\mapsto \{\eta\in\R^m\mid\of{B^{*}\bar{w}}^{\mathsf{T}}\eta\geq a\}\in\G_m$ is a weak*-closed convex function. Hence, the desired finiteness condition follows since Remark~\ref{closedhull} yields 
\begin{equation}
(\cl D_{\ell,r}^{\market})(0) \subseteq  \{\eta\in\R^m\mid\of{B^{*}\bar{w}}^{\mathsf{T}}\eta\geq a\}\neq\R^{m}.
\end{equation}

For the second part, \eqref{shortfallmaxdivergence} yields $R_{\ell}(X)\subseteq D_{\ell,r}(X)$ for every $X\in L_{m}^{\infty}$; hence, by Definition~\ref{closedhull}, $(\cl R^{\market}_{\ell})(Y)\subseteq(\cl D^{\market}_{\ell,r})(Y)$ for every $Y\in L_d^{\infty}$. The result follows now from the previous~part.
\end{proof}

\section*{Acknowledgments}

The authors are grateful to two anonymous referees whose comments were very helpful in improving the paper. The authors would like to thank Zachary Feinstein and Samuel Drapeau as well for useful comments on set-valued entropic risk measures and optimized certainty equivalents.


\begin{thebibliography}{9}
\bibitem[Ararat \& Rudloff(2016)]{AR15}
\c{C}.~Ararat \& B.~Rudloff (2016) Dual representations for systemic risk measures. arXiv:\href{http://arxiv.org/abs/1607.03430}{1607.03430}.

\bibitem[Astic \& Touzi(2007)]{AT07}
F.~Astic \& N.~Touzi (2007) No arbitrage conditions and liquidity,  \emph{Journal of Mathematical Economics} \textbf{43} (6), 692--708.

\bibitem[Barrieu \& El~Karoui(2008)]{barrieu}
P.~Barrieu \& N. El Karoui (2008) Pricing, hedging and optimally designing derivatives via minimization of risk measures. In: \emph{Volume on Indifference Pricing} (R.~Carmona, ed.), 77--146. Princeton, New Jersey: Princeton University Press.

\bibitem[Ben~Tahar \& Lepinette(2014)]{LBT13}
I.~Ben~Tahar \& E.~Lepinette (2014) Vector-valued coherent risk measure processes, \emph{International Journal of Theoretical and Applied Finance} \textbf{17} (2), 1450011.

\bibitem[Ben~Tahar(2006)]{B06}
I.~Ben~Tahar (2006) Tail conditional expectation for vector-valued risks, \emph{SFB 649 Discussion Papers} 2006-029, Humboldt University, Collaborative Research Center 649.

\bibitem[Ben~Tal \& Teboulle(1986)]{oldoce}
A.~Ben-Tal \& M.~Teboulle (1986) Expected utility, penalty functions and duality in stochastic nonlinear programming, \emph{Management Science} \textbf{32} (11), 1445--1466.

\bibitem[Ben~Tal \& Teboulle(2007)]{bt:oce}
A.~Ben-Tal \& M.~Teboulle (2007) An old-new concept of convex risk measures: the optimized certainty equivalent, \emph{Mathematical Finance} \textbf{17} (3), 449--476.

\bibitem[Biagini \emph{et~al.}(2015)]{BFFMb15}
F.~Biagini, J.-P.~Fouque, M.~Frittelli \& T.~Meyer-Brandis (2015) A unified approach to systemic risk measures via acceptance sets. arXiv:\href{http://arxiv.org/abs/1503.06354}{1503.06354}.

\bibitem[Borwein \& Lewis(1992)]{bl:pfcp}
J.~M.~Borwein \& A.~S.~Lewis (1992) Partially finite convex programming, Part I: Quasi relative interiors and duality theory, \emph{Mathematical Programming} \textbf{57} (1), 15--48.

\bibitem[Burgert \& R{\"u}schendorf(2006)]{BR06}
C.~Burgert \& L.~R{\"u}schendorf (2006) Consistent risk measures for portfolio vectors, \emph{Insurance: Mathematics and Economics} \textbf{38} (2), 289--297.

\bibitem[Campi \& Owen(2011)]{CampiOwen11}
L.~Campi \& M.~P.~Owen (2011) Multivariate utility maximization with proportional transaction costs, Finance and Stochastics \textbf{15} (3), 461--499.

\bibitem[Cascos \& Molchanov(2016)]{CM13}
I.~Cascos \& I.~Molchanov (2016) Multivariate risk measures: a constructive approach based on selections, \emph{Mathematical Finance} \textbf{26} (4), 867--900.

\bibitem[\c{C}etin \emph{et~al.}(2004)]{CJP04}
U.~\c{C}etin, R.~A.~Jarrow \& P.~Protter (2004) Liquidity risk and arbitrage pricing theory, \emph{Finance and Stochastics} \textbf{8} (3), 311--341.

\bibitem[\c{C}etin \& Rogers(2007)]{CR07}
U.~\c{C}etin \& L.~C.~G.~Rogers (2007) Modelling liquidity effects in discrete time, \emph{Mathematical Finance} \textbf{17} (1), 15--29.

\bibitem[Chen \emph{et~al.}(2013)]{CIM13}
C.~Chen, G.~Iyengar \& C.~Moallemi (2013) An axiomatic approach to systemic risk, \emph{Management Science} \textbf{59} (6), 1373--1388.

\bibitem[Cherny \& Kupper(2007)]{CKpre}
A.~Cherny \& M.~Kupper (2007) Divergence utilities. SSRN:\href{https://papers.ssrn.com/sol3/papers2.cfm?abstract_id=1023525}{1023525}.

\bibitem[Csisz\'{a}r(1967)]{csi}
I.~Csisz\'{a}r (1967) On topological properties of f-divergence, \emph{Studia Scientiarum Mathematicarum Hungarica} \textbf{2} (1), 329-–339.

\bibitem[Feinstein \& Rudloff(2013)]{FR13}
Z.~Feinstein \& B.~Rudloff (2013) Time consistency of dynamic risk measures in markets with transaction costs, \emph{Quantitative Finance} \textbf{13} (9), 1473--1489.

\bibitem[Feinstein \& Rudloff(2015a)]{FR13b}
Z.~Feinstein \& B.~Rudloff (2015a) Multi-portfolio time consistency for set-valued convex and coherent risk measures, \emph{Finance and Stochastics} \textbf{19} (1), 67–-107.

\bibitem[Feinstein \& Rudloff(2015b)]{FRsurvey}
Z.~Feinstein \& B.~Rudloff (2015b) A comparison of techniques for dynamic risk measures with transaction costs. In: \emph{Set Optimization and Applications - The State of the Art} (A.~H.~Hamel, F.~Heyde, A.~L\"{o}hne, B.~Rudloff, C.~Schrage, eds.), 3--41. Berlin, Heidelberg: Springer-Verlag.

\bibitem[Feinstein \emph{et~al.}(2017)]{FRW15}
Z.~Feinstein, B.~Rudloff \& S.~Weber (2017) Measures of systemic risk. Forthcoming in
SIAM Journal on Financial Mathematics, arXiv:\href{http://arxiv.org/abs/1502.07961}{1502.07961}.

\bibitem[F\"{o}llmer \& Schied(2002)]{fs:srm}
H.~F\"{o}llmer \& A.~Schied (2002) Convex measures of risk and trading constraints, \emph{Finance and Stochastics} \textbf{6} (4), 429--447.

\bibitem[F\"{o}llmer \& Schied(2011)]{fs:sf}
H.~F\"{o}llmer \& A.~Schied (2001) Stochastic finance: an introduction in discrete time, third revised and extended edition. Berlin, New York: Walter de Gruyter.

\bibitem[Hamel(2009)]{andreasduality}
A.~H.~Hamel (2009) A duality theory for set-valued functions I: Fenchel conjugation theory, \emph{Set-Valued and Variational Analysis} \textbf{17} (2), 153--182.

\bibitem[Hamel \& Heyde(2010)]{hh:duality}
A.~H.~Hamel \& F.~Heyde (2010) Duality for set-valued measures of risk, \emph{SIAM Journal on Financial Mathematics} \textbf{1} (1), 66--95.

\bibitem[Hamel \emph{et~al.}(2011)]{hhr:setval}
A.~H.~Hamel, F.~Heyde \& B.~Rudloff (2011) Set-valued risk measures for conical market models, \emph{Mathematics and Financial Economics} \textbf{5} (1), 1--28.

\bibitem[Hamel \& L\"{o}hne(2014)]{hl:lagrange}
A.~H.~Hamel \& A.~L\"{o}hne (2014) Lagrange duality in set optimization, \emph{Journal of Optimization Theory and Applications} \textbf{161} (2), 368--397.

\bibitem[Hamel \emph{et~al.}(2014)]{HLR13}
A.~H.~Hamel, A.~L\"{o}hne \& B.~Rudloff (2014) Benson type algorithms for linear vector optimization and applications, \emph{Journal of Global Optimization} \textbf{59} (4), 811--836.

\bibitem[Hamel \emph{et~al.}(2013)]{hry:avar}
A.~H.~Hamel, B.~Rudloff \& M.~Yankova (2013) Set-valued average value at risk and its computation, \emph{Mathematics and Financial Economics} \textbf{7} (2), 229--246.

\bibitem[Heyde \& L\"{o}hne(2011)]{heydeloehne11}
F.~Heyde \& A.~L\"{o}hne (2011) Solution concepts in vector optimization: a fresh look at an old story, \emph{Optimization} \textbf{60} (12), 1421-1440.

\bibitem[Jouini \& Kallal(1995)]{JK95}
E.~Jouini \& H.~Kallal (1995) Martingales and arbitrage in securities markets with transaction costs, \emph{Journal of Economic Theory} \textbf{66} (1), 178--197.

\bibitem[Jouini \emph{et~al.}(2004)]{jouini}
E.~Jouini, M.~Meddeb \& N.~Touzi (2004) Vector-valued coherent risk measures, \emph{Finance and Stochastics} \textbf{8} (4), 531--552.

\bibitem[L\"{o}hne \& Rudloff(2014)]{LR13}
A.~L{\"o}hne \& B.~Rudloff (2014) An algorithm for calculating the set of superhedging portfolios in markets with transaction costs, \emph{International Journal of Theoretical and Applied Finance} \textbf{17} (2), 1450012.

\bibitem[Kabanov(1999)]{kabanov}
Y.~M.~Kabanov (1999) Hedging and liquidation under transaction costs in currency markets, \emph{Finance and Stochastics} \textbf{3} (2), 237--248.

\bibitem[Pennanen \& Penner(2010)]{PP10}
T.~Pennanen \& I.~Penner (2010) Hedging of claims with physical delivery under convex transaction costs, \emph{SIAM Journal on Financial Mathematics} \textbf{1} (1), 158--178.

\bibitem[Rockafellar(1970)]{rockafellar}
R.~T.~Rockafellar (1970) Convex analysis. Princeton, New Jersey: Princeton University Press.

\bibitem[Rockafellar \& Wets(1998)]{rockafellar2}
R.~T.~Rockafellar \& R.~J.-B.~Wets (1998) Variational analysis. Berlin, Heidelberg: Springer-Verlag.

\bibitem[Rogers \& Singh(2010)]{RS10}
L.~C.~G.~Rogers \& S.~Singh (2010) The cost of illiquidity and its effects on hedging, \emph{Mathematical Finance} \textbf{20} (4), 597--615.

\bibitem[Schied(2007)]{schied}
A.~Schied (2007) Optimal investments for risk- and ambiguity-averse preferences: a duality approach, \emph{Finance and Stochastics} \textbf{11} (1), 107--129.

\bibitem[Sion(1958)]{sion}
M.~Sion (1958) On general minimax theorems, \emph{Pacific Journal of Mathematics} \textbf{8} (1), 171–-176.

\bibitem[Weber \emph{et~al.}(2013)]{Weber13}
S.~Weber, W.~Anderson, A.-M.~Hamm, T.~Knispel, M.~Liese \& T.~Salfeld (2013) Liquidity-adjusted risk measures, \emph{Mathematics and Financial Economics} \textbf{7} (1), 69--91.

\bibitem[Zalinescu(2002)]{zalinescu}
C.~Zalinescu (2002) Convex analysis in general vector spaces. Singapore: World Scientific.
\end{thebibliography}
\end{document}